\documentclass[a4paper,twocolumn,accepted=2024-06-12,11pt]{quantumarticle}
\pdfoutput=1

\usepackage{adjustbox}
\usepackage{amssymb}
\usepackage{graphicx}
\usepackage[english]{babel}
\usepackage{bbold}
\usepackage{braket}
\usepackage{graphicx}
\usepackage{float}
\usepackage{makeidx}
\usepackage{enumitem}
\usepackage{amsthm}
\usepackage{tikz}
\usepackage{color, colortbl}
\usetikzlibrary{quantikz}
\usepackage{url}
\usepackage[utf8]{inputenc}
\usepackage[caption=false]{subfig}
\usepackage{scrextend}
\usepackage[colorlinks=true,citecolor=red,linkcolor=brown,urlcolor=magenta]{hyperref}
\usepackage{makecell}
\usepackage{tabularx}
\usepackage{multirow}
\usepackage[]{algpseudocode}
\usepackage[numbers,sort&compress]{natbib}

\usepackage[outercaption]{sidecap}

\definecolor{green}{HTML}{66FF66}
\definecolor{myGreen}{HTML}{009900}

\newcounter{protocol}

\newcommand{\sbline}{\\[.5\normalbaselineskip]}

\newtheorem{theorem}{Theorem}
\newtheorem{lemma}{Lemma}

\DeclareMathAlphabet{\pazocal}{OMS}{zplm}{m}{n}
\newcommand{\E}{\pazocal{E}}

\newcommand{\R}{\pazocal{R}}
\newcommand{\h}{\pazocal{H}}

\newcommand{\C}{\pazocal{C}}
\newcommand{\p}{\pazocal{P}}

\newcommand{\D}{\pazocal{D}}

\newcommand{\Q}{\pazocal{Q}}

\newcommand{\U}{\pazocal{U}}

\newcommand{\Z}{\pazocal{Z}}
\newcommand{\X}{\pazocal{X}}
\newcommand{\Y}{\pazocal{Y}}
\newcommand{\I}{\pazocal{I}}

\begin{document}

\title{Efficiently improving the performance of noisy quantum computers}


\author{Samuele Ferracin}
\thanks{Now at IBM Quantum.}
\affiliation{Keysight Technologies Canada, Kanata, ON K2K 2W5, Canada
}
\affiliation{Department of Applied Mathematics, University of Waterloo, Waterloo, Ontario N2L 3G1, Canada}

\author{Akel Hashim}
\affiliation{Quantum Nanoelectronics Laboratory, Dept. of Physics, University of California at Berkeley, Berkeley, CA 94720, USA
}
\affiliation{Applied Math and Computational Research Division, Lawrence Berkeley National Lab, Berkeley, CA 94720, USA}

\author{Jean-Loup Ville}
\thanks{Now at Alice \& Bob.}
\affiliation{Quantum Nanoelectronics Laboratory, Dept. of Physics, University of California at Berkeley, Berkeley, CA 94720, USA
}

\author{Ravi Naik}
\affiliation{Quantum Nanoelectronics Laboratory, Dept. of Physics, University of California at Berkeley, Berkeley, CA 94720, USA
}
\affiliation{Applied Math and Computational Research Division, Lawrence Berkeley National Lab, Berkeley, CA 94720, USA}

\author{Arnaud Carignan-Dugas}
\affiliation{Keysight Technologies Canada, Kanata, ON K2K 2W5, Canada
}

\author{Hammam Qassim}
\affiliation{Keysight Technologies Canada, Kanata, ON K2K 2W5, Canada
}

\author{Alexis Morvan}
\thanks{Now at Google Quantum AI.}
\affiliation{Quantum Nanoelectronics Laboratory, Dept. of Physics, University of California at Berkeley, Berkeley, CA 94720, USA
}
\affiliation{Applied Math and Computational Research Division, Lawrence Berkeley National Lab, Berkeley, CA 94720, USA}

\author{David I. Santiago}
\affiliation{Quantum Nanoelectronics Laboratory, Dept. of Physics, University of California at Berkeley, Berkeley, CA 94720, USA
}
\affiliation{Applied Math and Computational Research Division, Lawrence Berkeley National Lab, Berkeley, CA 94720, USA}

\author{Irfan Siddiqi}
\affiliation{Quantum Nanoelectronics Laboratory, Dept. of Physics, University of California at Berkeley, Berkeley, CA 94720, USA
}
\affiliation{Applied Math and Computational Research Division, Lawrence Berkeley National Lab, Berkeley, CA 94720, USA}
\affiliation{Materials Sciences Division, Lawrence Berkeley National Lab, Berkeley, CA 94720, USA}

\author{Joel J. Wallman}
\affiliation{Keysight Technologies Canada, Kanata, ON K2K 2W5, Canada
}
\affiliation{Department of Applied Mathematics, University of Waterloo, Waterloo, Ontario N2L 3G1, Canada}
\thanks{Now at ORCA Computing.}

\begin{abstract}
Using near-term quantum computers to achieve a quantum advantage requires efficient strategies to improve the performance of the noisy quantum  devices presently available. We develop and experimentally validate two efficient error mitigation protocols named ``Noiseless Output Extrapolation" and ``Pauli Error Cancellation" that can drastically enhance the performance of quantum circuits composed of noisy cycles of gates. By combining popular mitigation strategies such as probabilistic error cancellation and noise amplification with efficient noise reconstruction methods, our protocols can mitigate a wide range of noise processes that do not satisfy the assumptions underlying existing mitigation protocols, including non-local and gate-dependent processes. We test our protocols on a four-qubit superconducting processor at the Advanced Quantum Testbed. We observe  significant improvements in the performance of both structured and random circuits, with up to $86\%$ improvement in variation distance over the unmitigated outputs. Our experiments demonstrate the effectiveness of our protocols, as well as their practicality for current hardware platforms.
\end{abstract}

\noindent\footnotesize{}

\noindent\footnotesize{S. Ferracin and A. Hashim contributed equally.}

\maketitle

\section{Introduction}
The last few years have seen unprecedented advances in quantum technologies, particularly in the development of Noisy Intermediate-Scale Quantum (NISQ)  devices.
These devices can already outperform their classical counterparts in specific tasks~\cite{GoogleSupremacy19}, but suffer from significant noise that corrupts their outputs. As fault tolerance is not expected to become available in the immediate future
\cite{Preskill18}, understanding how to optimize the performance of NISQ  devices is of paramount importance.

In recent years, much effort has been devoted to finding alternatives to fault tolerance that are feasible in the near term. This paved the way for the development of a rich set of error mitigation~(EM) protocols~\cite{LB17,TBG17,EBL18,KandalaEtal19,SQCBY20,GTHLMZ20,LRMKSZ20,ECBY21,KWYMGTK21,Koczor21,HNdYB20,HMOLRBWBM21,BHdJNP21}. Unlike fault-tolerant protocols, EM protocols do not attempt to correct errors that occur in noisy circuits. On the contrary, by actively amplifying noise in a controlled way and post-processing the outputs, they eventually extrapolate correct outputs from the noisy ones. EM protocols require a high number of samples and are generally inefficient in the circuit size~\cite{TEMG21}. On the other hand, they typically have little overhead in qubits and gates, which makes them practical for today's  devices.

Despite the promising results obtained in a large number of experimental demonstrations~\cite{EBL18,KandalaEtal19,SongEtAl19,ZhangEtAL20,SQCBY20,CACC20,HNdYB20,BHdJNP21,KWYMGTK21}, performing EM on large circuits remains a challenging task. Some of the leading proposals~\cite{TBG17,EBL18} require reconstructing the noise afflicting multi-qubit operations using Gate-Set Tomography~(GST)~\cite{BKEtal13,DanG15,BK&al17}. Since GST is inefficient in the number of qubits, the noise reconstruction is typically performed by analysing individual one- and two-qubit gates while ignoring the rest of the system. This approach severely limits the effectiveness of EM for large circuits, where noise processes involving more than two qubits (such as crosstalk) typically play a major role~\cite{MCWG20,HFW19,Hashim20}. 
For this reason, so far the EM protocols that rely on GST have been demonstrated on circuits containing up to two qubits \cite{SongEtAl19,ZhangEtAL20}, where this limitation is irrelevant. Alternative protocols~\cite{TBG17,KandalaEtal19,LRMKSZ20,HNdYB20,GTHLMZ20,BHdJNP21,Cai21} that do not rely on noise reconstruction have been applied to larger circuits~\cite{KWYMGTK21}, but their effectiveness has been formally proven only for specific noise models, e.g. for depolarizing noise \cite{GTHLMZ20,LRMKSZ20,BHdJNP21}. Overall, further work is required before EM protocols can become of use for circuits of interesting sizes.

Recently, a number of protocols have been developed that can characterize multi-qubit noise processes more efficiently than GST \cite{FW19,HFW19,Flammia21}, under a set of realistic (and verifiable) assumptions of the noise. Among these is ``Cycle Error Reconstruction'' (CER), which can accurately reconstruct Pauli channels even when they act on more than two qubits~\cite{Hashim20}. Leveraging CER, in this paper we develop a novel approach to EM that targets ``cycles'' of gates (i.e., groups of gates applied in parallel to disjoint subsets of qubits~\cite{WE16}) rather than individual gates. This approach enables us to design two new EM protocols to efficiently mitigate noise that involves an arbitrary number of qubits|potentially up to an entire register. By targeting the noise afflicting full cycles of gates rather than individual gates, our protocols are fully robust to complex error processes (such as cross-talk, correlated errors, and non-depolarizing noise) that are known to affect present devices~\cite{MCWG20,HFW19,Hashim20}, and that may compromise the effectiveness of existing EM protocols.

\begin{figure*}[!t]
     \centering
      \subfloat[][]{\includegraphics[clip,width=0.48\columnwidth]{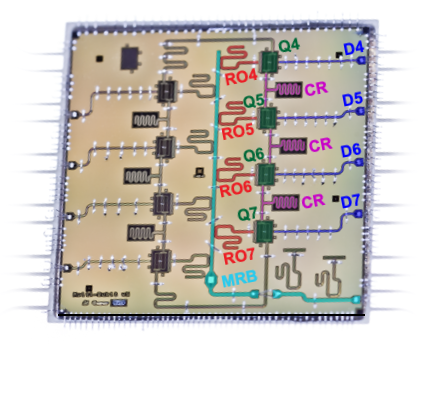}\label{fig:chip}}
      \qquad
      \subfloat[][]{\includegraphics[clip,width=1.48\columnwidth]{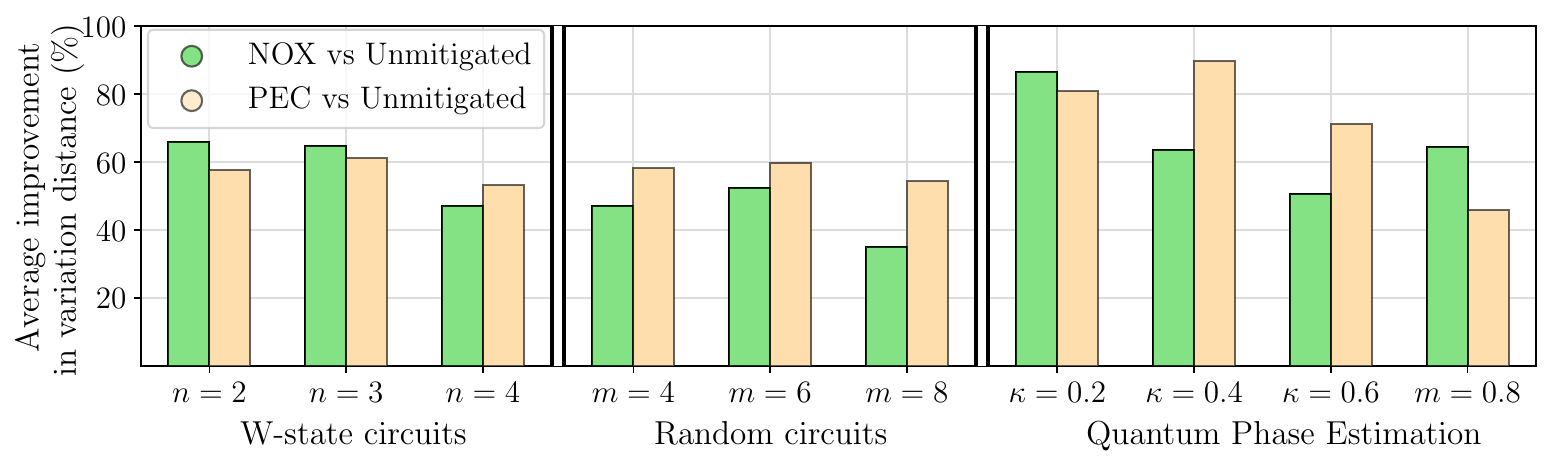}\label{fig:improvements}}
     \caption{\small (a) Micrograph of the superconducting quantum processor at the Advanced Quantum Testbed (reprinted with permission from Ref.~\cite{Hashim20}). In this work, we used the four transmon qubits (green) with independent drive lines (blue) out of a total of eight qubits arranged in a ring geometry---the other four qubits are inactive on the device. The qubits are coupled to nearest neighbors via coupling resonators (CR, purple), and are dispersively measured via independent readout resonators (red) coupled to a multiplexed readout bus (MRB, cyan). (b) Average improvements in variation distance obtained in successful implementations of NOX and PEC protocols. In our experiments we apply PEC and NOX to four-qubit random circuits of varying depth $m$ as well as to structured circuits, such as circuits to prepare W states with $n=2,3,4$ qubits (Eq.~\ref{eq:w_states}) and to estimate a parameter $\kappa$ through the quantum phase estimation algorithm. We compute the variation distance (VD, Eq.~\ref{eq:def_VD}) between the ideal and experimental probability distributions of the outputs, and we quantify the improvement in VD under EM as $1-\textrm{D}_\textrm{EM}/\textrm{D}_\textrm{unm}$, where $\textrm{D}_\textrm{EM}$ ($\textrm{D}_\textrm{unm}$) is the VD for the mitigated (unmitigated) outputs. We observe drastic improvements in VD for both NOX and PEC, ranging from $32\%$ to as high as $86\%$.}
     \label{fig:chip_and_summary}
\end{figure*}

\newcolumntype{g}{>{\columncolor{gray!10}}c}
\newcolumntype{y}{>{\columncolor{yellow!10}}c}
\begin{table}[t]
\footnotesize
    \setlength{\extrarowheight}{4pt}
    \begin{tabular}{r|| g | y}
   & \textbf{PEC} & \textbf{NOX} \\
    \hline
    \hline
   &&\\[-2ex]
   \textbf{R.time} & $\:\dfrac{(1-n\varepsilon)^{-2m}}{\sigma^2}\:$ & $\dfrac{m^3}{\sigma^2}$ \\
   &&\\[-2ex]
   \textbf{Bias} & $\delta_\textrm{PEC}+\delta_\textrm{rec},$ with & $\delta_\textrm{NOX}+\delta_\textrm{rec},$ with \\
   & $\delta_\textrm{PEC}=O( mn^2\varepsilon^2)$ & $\delta_\textrm{NOX}=O( m^2n^2\varepsilon^2)$\\
\end{tabular}
    \caption{Runtime and bias for our EM protocols, calculated as functions of the desired standard deviation $\sigma$ of the results, the circuit depth $m$, and the error rate of each cycle. (For simplicity, in this table we assume that every cycle has an error rate equal to $n\varepsilon$, where $n$ is the number of qubits in the input circuit and $\varepsilon$ is a constant; see theorems \ref{th:PEC} and \ref{th:ce} in section \ref{sec:our_protocols} for a generalisation.) The quantity $\delta_\textup{rec}$ depends on the accuracy of the noise reconstruction, with $\delta_\textup{rec}=0$ if the noise is known exactly (cfr.~lemmas \ref{lem:relax_PEC} and \ref{lem:relax_ce}). In comparison, an unmitigated implementation of the input circuit hes runtime ${1}/{\sigma^2}$ and bias $O( mn\varepsilon)$.
    }
    \label{tab:summary}
\end{table}

\subsection{Our protocols in summary}

Our protocols (which we name ``Pauli Error Cancellation’’, or PEC, and ``Noiseless Output Extrapolation’’, or NOX) take as input a quantum circuit and an operator $O$, and return an estimator of the expectation value of $O$ at the end of a noiseless implementation of the input circuit. To compute this estimator they undertake two different approaches. In particular, PEC is built around quasi-probabilistic error cancellation, one of the most popular techniques in EM~\cite{LB17,TBG17,EBL18,KandalaEtal19}, while NOX requires amplifying the noise afflicting individual cycles. Despite this difference, they require performing the same fundamental tasks: characterizing the noise in the input circuit with CER, implementing a set of noisy circuits, and finally post-processing their results.

The estimators returned by NOX and PEC have a residual bias that depends on the noise levels of the input circuit, as well as on the accuracy of the noise reconstruction (Table \ref{tab:summary}). Crucially, if backed by an accurate noise reconstruction, this bias is quadratic in the error rate of the input (unmitigated) circuit. This is a significant improvement over the unmitigated estimators, whose biases are linear in the error rates.

While being able to mitigate a broader class of noise processes than other existing EM protocols, NOX and PEC require a similar runtime as the available protocols. In particular, if the input circuit is afflicted by moderate noise (i.e. in the notation of Table~\ref{tab:summary}, if $\varepsilon<1/mn$), they remain \textit{efficient}, meaning that their runtimes scale polynomially with the desired statistical accuracy of the results.

\subsection{Our experiments in summary}
We demonstrate our EM protocols on four fixed-frequency superconducting transmon qubits (labeled Q4, Q5, Q6, and Q7) on an eight qubit quantum processor ($\texttt{AQT@LBNL Trailblazer8-v5.c2}$; see Fig.~\ref{fig:chip}) at the Advanced Quantum Testbed at Lawrence Berkeley National Lab \cite{AQT}. Single-qubit gates are implemented using resonant Rabi-driven $X_{\pi/2}$ gates and virtual phase shifts between pulses \cite{mckay2017efficient}. The native two-qubit gate on the device is a controlled-$Z$ ($cZ$) implemented via off-resonant drives between neighboring qubits \cite{MitchellAndOthers21}. We apply PEC and NOX to circuits of different size and nature. We implement each circuit multiple times, with and without EM, and calculate the average variation distance~(Eq.~\ref{eq:def_VD}) between ideal and experimental probability distributions of the outputs. We observe significant improvements in variation distance under EM for every circuit (Fig.~\ref{fig:improvements}).

The $cZ$ gates are the noisiest components in our circuits. However, the CER data show that when a $cZ$ gate is performed in parallel with idling spectator qubits, the idling qubits experience higher levels of noise than the qubits being entangled (see Appendix). In this scenario, performing mitigation at the level of individual gates (e.g.~with the gate-centred EM protocols that rely on GST~\cite{TBG17,EBL18}) would not lead to significant improvements in the outputs, and could in fact amplify the noise acting on the idling qubits. On the contrary, by targeting cycles rather than gates, NOX and PEC are able to provide the visible improvements reported in Fig~\ref{fig:improvements}.\\

This paper is organized as follows. In section \ref{sec:background} we define our notation, list the assumptions made throughout the paper and provide a brief overview of CER. In section~\ref{sec:our_protocols} we describe our protocols and state their main properties. In section \ref{sec:experiment} we describe our experiments.

\section{Background}
\label{sec:background}

\subsection{Notation and assumptions}
\label{sec:notation_and_assumptions}
We denote unitary gates with capital Latin letters and Completely Positive Trace-Preserving (CPTP) maps with calligraphic letters. We write $\U=\{U_l\}_{l=1}^L$ to indicate that $\U$ has Kraus operators $U_l$. We use $\circ$ to indicate the composition of maps, e.g.~$\circ_{j=1}^m\U_j=\U_m\cdots\U_1$.

To prove our results we make the following two assumptions:
\begin{itemize}[leftmargin=6mm]
    \item[A1.] We assume that the noise is Markovian and time-stationary, i.e., that a noisy implementation of an operation $\U$ can be written as $\D_\U\U$, where $\D_\U$ is a (potentially operation-dependent) CPTP map that is fixed in time.
    \item[A2.] We assume that the cycles of one-qubit gates suffer gate-independent noise, i.e., $\D_\U=\D$ for all the cycles of one-qubit gates~$\U$.
\end{itemize}
These are standard assumptions in the literature on noise characterisation and mitigation \cite{BKEtal13,DanG15,BK&al17,WHEB04,MGSPC13,WE16,Erhard&al19,HFW19,HFW19,FKD18,FMMD21,BKEtal13,DanG15,BK&al17,WHEB04,MGSPC13,MGE11,CWE15} and can be relaxed at the cost of more complex notation \cite{WE16}.

In addition to A1 and A2, we assume that every noise process $\D_\U$ in our circuits is a Pauli channel, i.e., that it maps an $n$-qubit state $\rho$ into \begin{equation}
\label{eq:pauli_error_rates}
    \D_\U(\rho)=\sum_{k=0}^{4^n-1}\epsilon_k^{(\U)}\p_k(\rho) \:.
\end{equation}
Here, the ``Pauli errors'' $\p_k\in\{\I,\X,\Y,\Z\}^{\otimes n}$ are $n$-qubit Pauli operators and the ``Pauli error rates'' $\epsilon_k^{(\U)}$ are probabilities (we set $\p_0=\I^{\otimes n}$ for convenience). Although not every noise process is a Pauli channel, under the assumptions A1 and A2 every process can be efficiently transformed into Pauli channels via Randomized Compiling~\cite{WE16,Hashim20}, available on \texttt{True-Q}~\cite{trueq}.

\subsection{Cycle Error Reconstruction}
\label{sec:CER}
CER (available on the software \texttt{True-Q} \cite{trueq}) is a protocol to efficiently characterise noisy cycles with high accuracy. In more detail, let $\D _\h\h$ be a noisy implementation of an $n$-qubit Clifford cycle $\h$ with Pauli noise $\D_\h$. In its simplest form, CER takes as input the cycle $\h$ and a positive integer $K\leq n$. After characterizing the cycle's noise via Cycle Benchmarking \cite{Erhard&al19} and post-processing the results of these benchmarking circuits \cite{ Hashim20},  it estimates the Pauli error rates associated to \textit{all} the errors of weight $K$|that is, to all the errors that affect up to $K$ qubits simultaneously.

The accuracy of these estimates depends on the nature of the state-preparation and measurement~(SPAM) errors afflicting the  device in use. If no assumption is made on the SPAM errors, the estimates are
averages over small subsets of error rates that typically contain up to two elements. On the contrary, if state-preparation errors are negligible compared to measurement errors or  vice-versa (which is the case for many of today's platforms~\cite{FMMD21,MZO20} and is routinely assumed in related works~\cite{BSKMG20,Flammia21,MBZO21}), CER can estimate all the error rates individually.

Note that the total number of weight-$K$ errors grows as~$n^{K}$, hence CER is not efficient in $K$. However, at fixed $K$, CER scales polynomially in $n$. This means that CER can efficiently perform an accurate characterisation of the noisy cycle of interest, provided that low-weight errors encompass the majority of the probability distribution. This is often the case on state-of-the-art  devices, where high-weight errors ($K\geq3$) occur with negligible probability|as an example, see the CER data in Fig.~\ref{fig:KNR} or the CER data reported in Ref.~\cite{Hashim20}.

\section{Our EM protocols}
\label{sec:our_protocols}
In this section we describe our EM protocols and their overheads and biases. Without loss of generality, we consider input circuits that alternate cycles of one-qubit gates and cycles of Clifford two-qubit gates, implementing operations of the form
\begin{equation}
\label{eq:input}
    \C=\E_{m+1}\h_m\E_m\cdots\E_2\h_1\E_1=\E_{m+1}\big(\circ_{j=1}^m\h_j\E_j\big)\:.
\end{equation}
Here, $\E_j$ (respectively $\h_j$) is the operation implemented by the  $j$th cycle of one-qubit gates (respectively by the  $j$th cycle of two-qubit gates). Under the assumptions A1 and A2 (section \ref{sec:notation_and_assumptions}), a noisy implementation of the input circuit performs the map 
\begin{equation}
\label{eq:noisy_input}
\widetilde{\C}=\E_{m+1}\big(\circ_{j=1}^m\D_{\h_j}\h_j\E_j\big)\:,
\end{equation}
where we have recompiled the gate-independent noise afflicting the cycles of one-qubit gates into that afflicting the cycles of two-qubit gates and obtained the Pauli channels $\D_{\h_j}$ via Randomized Compiling.
Motivated by the above equation, for convenience we will refer to the cycles of one-qubit gates as ``noiseless under compilation'' (or simply as ``noiseless'') and to all the other cycles as ``noisy''.

\subsection{Pauli Error Cancellation}
\label{sec:PEC}
We begin by explaining the main ideas behind quasi-probabilistic error cancellation, one of the primary ingredients employed by our PEC protocol. Quasi-probabilistic error cancellation is a strategy to compute an unbiased estimator by sampling from a distribution of biased estimators~\cite{LB17,TBG17}. Formally, let $\U$ be a desired, noiseless operation, and let  $\{\widetilde{\U}_l\}_{l=1}^{L}$ be a set of noisy operations that can be implemented experimentally. The task of quasi-probabilistic error cancellation is to calculate a set of probabilities $q_l$, a set of signs $s_l\in\{-1,+1\}$ and a number $C_\textup{tot}>0$ (called the ``cost'') such that for any state $\rho_\textup{in}$ and operator~$O$,
\begin{align}
\label{eq:QPEC}
C_\textup{tot}\sum_{l=1}^Ls_lq_l \textup{Tr}\big[O\widetilde{\U}_{l}(\rho_\textup{in})\big]=\:\textup{Tr}\big[O\U(\rho_\textup{in})\big]+\delta\:,
\end{align}
where $\delta\approx0$ represents a residual bias and captures the effectiveness of the EM protocol. All the EM protocols based on quasi-probabilistic error cancellation guarantee a negligible bias $\delta\approx0$, provided that the noisy maps $\widetilde{\U}_l$ can be accurately characterised.

We can now present our PEC protocol, which is formally described in section I of the Supplementary Material. PEC takes as input the circuit $\C$ (Eq.~\ref{eq:input}), the Pauli error rates $\{\epsilon_k^{(\h_j)}\}$ of all the noisy cycles $\h_j$ in $\C$ (which are computed in advance with CER), an $n$-qubit state $\rho_\textup{in}$, an operator $O$ such that the spectral norm $||O||_\infty\sim 1$, and a number $\sigma\in(0,1)$ representing the desired standard deviation of the results. It uses quasi-probabilistic error cancellation to suppress the noise afflicting the  noisy cycles $\h_j$, and eventually it returns an estimator $\widehat{E}_\textup{PEC}(O)$ of Tr$\big[O\C\big(\rho_\textup{in}\big)\big]$. To calculate $\widehat{E}_\textup{PEC}(O)$, PEC requires running $N=(C_\textup{tot}/\sigma)^2$ circuits in total, with cost given by
\begin{equation}
\label{eq:cost}
C_{\textup{tot}}=\prod_{j=1}^{m}\frac{1}{\big(\epsilon_0^{(\h_j)}\big)^2
-\sum_{k=1}^{4^n}\big(\epsilon_k^{(\h_j)}\big)^2}\:.
\end{equation}
Each of these circuits is obtained by appending randomly chosen Pauli gates to the noiseless cycles. Specifically, every circuit in PEC implements an operation of the type
\begin{align}
\label{eq:circuit_PEC}
    \C^{\textup{(PEC)}}(\p_1,\ldots,\p_m)=\E_{m+1}\big(\circ_{j=1}^m\p_j\h_j\E_j\big)\:,
\end{align}
where $\p_{j}\in\{\I,\X,\Y,\Z\}^{\otimes n}$ is chosen at random with probability $\epsilon_{k_j}^{(\h_j)}$.Together with the $N$ circuits, PEC also initialises a list of signs $s_1,\ldots,s_N$, where $s_{k}=1$ if circuit $k$ contains an even number of random Pauli cycles $\p_j$ that are different from the identity and $s_k=-1$ otherwise. 

After initializing circuits and signs, PEC applies  Randomized Compiling to every circuit, runs the circuits and stores the results $r_1,\ldots,r_N$. Finally, it computes the estimator $\widehat{E}_\textup{PEC}(O)$ as
\begin{equation}
\widehat{E}_\textup{PEC}(O)=C_{\textup{tot}}\sum_{k=1}^N\frac{s_kr_k}{N}\:.
\end{equation}
The following theorem states the standard deviation and bias of $\widehat{E}_\textup{PEC}(O)$, under the simplifying assumption that the Pauli error rates of every noisy cycle are known exactly:
\begin{theorem}
\label{th:PEC}
 (Proof in section I of Supplementary Material).
Let the Pauli error rates of every noisy cycle $\h_j$ be known exactly. Under the assumptions A1 and A2 (section \ref{sec:notation_and_assumptions}), the number $\widehat{E}_\textup{PEC}(O)$ returned by our PEC protocol is an estimator of ${E}(O)$ with standard deviation $O(\sigma)$ and bias
\begin{align}
\delta_\textup{PEC}=O\bigg(C_\textup{tot}\sum_{j=1}^{m}(1-\epsilon_0^{(\h_j)})^2\bigg)\:.
\end{align}
\end{theorem}
Assuming perfect knowledge of the Pauli error rates is unrealistic for two reasons (more details in section \ref{sec:background}): Firstly, estimating \textit{all} the $4^n$ error rates of an $n$-qubit cycle is impractical even for few-qubit cycles, so we can only learn a few of them (e.g. the largest ones). Secondly, the estimates returned by CER are subject to statistical fluctuations. To relax this assumption, we show that our PEC protocol is robust to inaccuracies in the estimates of the Pauli error rates, provided that they are suitably small. Formally:
\begin{lemma}
\label{lem:relax_PEC}(Proof in section III of Supplementary Material). 
Let $\epsilon_l^{(\h_j)}$ be the Pauli error rates of the noisy cycle $\h_j$ and let $\widehat{\epsilon}_l^{(\h_j)}$ be the estimates computed with CER. Under the assumptions A1 and A2, the estimator $\widehat{E}_\textup{PEC}(O)$ returned by our PEC protocol has bias $\delta_\textup{PEC}'=\delta_\textup{PEC}+\delta_\textup{rec}$, where $\delta_\textup{PEC}$ is the bias in theorem~\ref{th:PEC} and
\begin{equation}
    \delta_\textup{rec}=O\bigg(\sum_{j=1}^m\sum_{l=0}^{4^n-1}\big|{\epsilon}_l^{(\h_j)}-\widehat{\epsilon}_l^{(\h_j)}\big|\bigg)\:.
\end{equation}
\end{lemma}
To better quantify the bias of PEC, let us assume for simplicity that the error probability is the same for every noisy cycle, i.e., $1-\epsilon_0^{(\h_j)}=\varepsilon$ for all $j\in\{1,\ldots,m\}$. In this case, the bias $\delta_\textup{PEC}$ in theorem~\ref{th:PEC} grows quadratically in $\varepsilon$ as $\delta_\textup{PEC}=O(m\varepsilon^2)$. 
Hence, if the Pauli error rates are known perfectly, PEC can successfully improve the performance of circuits with depth $m\lesssim\varepsilon^{-2}$. More generally, if we assume a fixed relative precision of the Pauli error rates, i.e. $\big|{\epsilon}_l^{(\h_j)}-\widehat{\epsilon}_l^{(\h_j)}\big|/\epsilon_l^{(\h_j)} = \beta $ for $\beta \ll 1$, 
we get $\delta_\textup{rec} = O(m \beta \epsilon)$. CER inherently provides Pauli error rates with multiplicative precision, and the relative uncertainty $\beta$ can be brought closer to zero by improving the quality of the characterization data set (e.g. by increasing the number of CER circuits and the number of shots). 
Overall, performing an accurate characterization is vital since a high relative uncertainty 
on the Pauli error rates might negatively impact the residual bias $\delta_\textup{PEC}'$.

To conclude the section we analyse the complexity of PEC. To achieve a fixed standard deviation $O(\sigma)$, PEC requires implementing $C_\textup{tot}^2/\sigma^2$ circuits, where $C_\textup{tot}$ typically grows exponentially with $m$. For example, when all the noisy cycles have the same error probability $\varepsilon$, we have $C_\textup{tot}=O((1-\varepsilon)^{-2m})$. Thus, in general PEC (as well as all the other protocols based on quasi-probabilistic cancellation \cite{TEMG21}) is inefficient due to the exponential scaling of the cost with the circuit depth. Nevertheless, if applied to circuits with depth $m\lesssim\varepsilon^{-2}$, PEC remains an efficient and practical solution.

\subsection{Noiseless Output Extrapolation}
\label{sec:ce}
NOX relies on the ability to amplify the noise afflicting individual  noisy cycles in the input circuit. Specifically, it requires replacing the noisy operations $\D_{\h_j}\h_j$ with $\D_{\h_j}^{\alpha}\h_j$ for integers $\alpha>1$. We begin this subsection by explaining how this amplification may be performed, then we describe NOX.

The traditional method to amplify the noise is the so-called ``Identity Insertion''~\cite{HNdYB20}, which consists of replacing a noisy cycle $\h_j$ with $\h_j(\h_j\h_j^{-1})^\alpha$. This method is efficient and is used by a number of other EM protocols~\cite{GTHLMZ20,LRMKSZ20,HNdYB20,KWYMGTK21,BHdJNP21}, but it is accurate only if two conditions are satisfied: Firstly, if $\h_j$ and $\h_j^{-1}$ are afflicted by identical noise. Secondly, if the noise and the cycle commute, i.e., $\D_{\h_j}\h_j=\h_j\D_{\h_j}$. The first condition is trivially satisfied by cycles for which $\h_j=\h_j^{-1}$, for example, by cycles containing a combination of $cZ$ and $cX$ gates. The second condition is satisfied by specific noise processes, e.g.~by the $n$-qubit depolarising channel, but not in general \cite{KWYMGTK21}. Importantly, CER allows checking if these two conditions are satisfied and to evaluate the accuracy of Identity Insertion before employing it in an experiment.

While we do not attempt to improve Identity Insertion, we propose an alternative method that can correctly amplify arbitrary noise processes. Our method (which we call ``Append Errors'') takes as input the circuit $\C$, a label $j\in\{1,\ldots,m\}$, the Pauli error rates $\epsilon_k^{(\h_j)}$ of the $j$th noisy cycle and the amplification factor $\alpha>1$. It returns the circuit
\begin{align}
\label{eq:append_errors}
    &\C'(j;\:k_1,\ldots,k_{\alpha-1})=\\
    &\E_{m+1}\big(\circ_{k=j+1}^m\h_k\E_k\big)\Q_{k_{\alpha-1}}\cdots\Q_{k_1}\h_j\E_j\big(\circ_{k=1}^{j-1}\h_k\E_k\big)\:,\nonumber
\end{align}
where each $\Q_{k_l}\in\{\I,\X,\Y,\Z\}^{\otimes n}$ is an $n$-qubit Pauli operator chosen at random with probability $\epsilon^{(\h_j)}_{k_l}$. Since $\sum_{k=0}^{4^n-1}\epsilon_k^{(\h_j)}\Q_k=\D_{\h_j}$, on average a noisy implementation of $\C'(j;\:k_1,\ldots,k_{\alpha-1})$ performs the operation
\begin{align}
\label{eq:circ_amplified_j}
    \widetilde{\C}'_{\h_j,\alpha}=
    &\sum_{k_1,\ldots,k_{\alpha-1}}\epsilon^{(\h_j)}_{k_{\alpha-1}}\cdots\epsilon^{(\h_j)}_{k_{1}}\:\C'(j;\:k_1,\ldots,k_{\alpha-1})\\
    =& \; \E_{m+1}\D_{\h_m}\h_m\E_m\cdots\big(\D_{\h_j} \big)^{\alpha}\:\h_j\cdots\h_1\E_1\:.\nonumber
\end{align}
This corresponds to the operation $\widetilde{\C}$ implemented by the noisy circuit except for the noise on the $j$th  noisy cycle, which is amplified by a factor $\alpha$.

We can now move onto presenting NOX. NOX (formally described section II of Supplementary Material) takes as input the circuit $\C$, an $n$-qubit state $\rho_\textup{in}$, an operator $O$ such that $||O||_\infty\sim 1$, a number $\sigma\in(0,1)$ representing the desired standard deviation of the results, an integer $\alpha>1$ and a Boolean $\textsf{id\_insert}\in\{\textsf{True},\textsf{False}\}$. It requires running $m+1$ circuits in total. The first of these circuits is identical to the input circuit, while the other $m$ circuits contain one  noisy cycle with noise amplified by a factor $\alpha$. If $\textsf{id\_insert}=\textsf{True}$, the noise amplification is performed with Identity Insertion, otherwise it is performed with Append Errors. Each circuit is implemented $m^2/(\alpha-1)^2\sigma^2$ times and yields a noisy estimator of $E(O)$. We denote with $\widetilde{E}_\textup{in}(O)$ the noisy estimator returned by the circuit that is identical to the input one, and by $\widetilde{E}_{\h_j,\alpha}(O)$ that returned by the circuit with amplified noise on the $j$th  noisy cycle. After running all the $m+1$ circuits, NOX returns the quantity
\begin{equation}
    \widehat{E}_\textup{NOX}(O)=\widetilde{E}_\textup{in}(O) + \sum_{j=1}^{m} \frac{\widetilde{E}_\textup{in}(O) - \widetilde{E}_{\h_j,\alpha}(O)}{\alpha - 1}\:.
\end{equation}
This quantity is yet another estimator of $E(O)$, but it is significantly more accurate than the noisy estimators. The following theorem states its standard deviation and bias, under the assumption that the noise amplification is performed exactly:
\begin{theorem}
\label{th:ce}
(Proof in section II of Supplementary Material). Let us assume that the noise of every  noisy cycle $\h_j$ in the input circuit can be amplified exactly by a factor $\alpha>1$. Under the assumptions A1 and A2 (section \ref{sec:notation_and_assumptions}), the number $\widehat{E}_\textup{NOX}(O)$ returned by our NOX protocol is an estimator of $E(O)$ with standard deviation $O(\sigma)$ and bias
\begin{align}
\delta_\textup{NOX}=O\bigg(\alpha\sum_{j=1}^{m}(1-\epsilon_0^{(\h_j)})\sum_{l=j}^{m}(1-\epsilon_0^{(\h_l)})\bigg)\:.
\end{align}
\end{theorem}
Overall, while the biases of the noisy estimators $\widetilde{E}_\textup{in}(O)$ and $\widetilde{E}_{\h_j,\alpha}(O)$ grow linearly with the cycles' error rates, the bias of $\widehat{E}_\textup{NOX}(O)$ only grows quadratically.
Note that the bias $\delta_\textup{NOX}$ also grows linearly with $\alpha$. Thus, choosing small values of $\alpha$ leads to better performance for NOX.

Assuming that the noise can be amplified exactly is unrealistic, both for Identity Insertion (since the noise may commute approximately but not exactly) and for Append Errors (since inevitable inaccuracies in the estimation of the Pauli error rates may lead to an imperfect amplification). To relax this assumption we prove the following lemma:

\begin{lemma}
\label{lem:relax_ce}
(Proof in section III of Supplementary Material). Let $\D_{\h_j}^\alpha\h_j$ be a noisy implementation of $\h_j$ with noise amplified exactly by a factor $\alpha$, and let $\widetilde{\R}_{\h_j}\D_{\h_j}\h_j$ be an implementation of $\h_j$ with noise amplified imperfectly. Under assumptions A1 and A2, the estimator $\widehat{E}_\textup{NOX}(O)$ returned by our NOX protocol has bias $\delta_\textup{NOX}'=\delta_\textup{NOX}+\delta_\textup{rec}$, where $\delta_\textup{NOX}$ is the bias in theorem \ref{th:ce} and
\begin{equation}
    \delta_\textup{rec}=O\bigg(\sum_{j=1}^m||\widetilde{\R}_{\h_j}-\D_{\h_j}^{\alpha-1}||_\diamond\bigg)\:.
\end{equation}
\end{lemma}
The above lemma is analogous to lemma \ref{lem:relax_PEC} for NOX, as it proves that an accurate but imperfect (i.e. a \textit{realistic}) noise amplification can still guarantee a high performance of our PEC protocol. Overall, if the noise is amplified perfectly, the bias of NOX grows quadratically in the cycles' error rate $\varepsilon$ as $\delta_\textup{NOX}=O(m^2\varepsilon^2)$, where for simplicity we assume that every cycle has the same error probability $\varepsilon$. If the noise is amplified imperfectly, we can expect $\delta_\textrm{rec}$ to grow linearly in $\varepsilon$ as $\delta_\textrm{rec}=O(\gamma m \varepsilon)$, where $\gamma$ is proportional to the inaccuracy in the noise amplification. Therefore, performing an accurate noise amplification is vital to ensure that the residual bias $\delta_\textup{NOX}'$ remains as low as possible.

Unlike PEC, NOX requires running a number of circuits that does not depend on the cycles' error rate. In particular, to achieve the desired standard deviation, NOX requires initialising $m+1$ circuits and running each of them $O(m^2)$ times. Thus, NOX has runtime $O(m^3)$ and is efficient in $m$. This result may seem to contradict Ref.~\cite{TEMG21}, which shows that the EM protocols are fundamentally inefficient in the circuit depth. However, Ref.~\cite{TEMG21} only considers protocols that have a fixed bias, independent of the circuit depth, while the bias of NOX grows quadratically with $m$.

We conclude this section by clarifying the differences between NOX and the existing protocols based on noise amplification. NOX can be seen as a noise-aware generalisation of the ``Random Identity Insertion Method''~(RIIM) presented in Ref.s~\cite{HNdYB20}, which is built around Identity Insertion. Even though NOX and RIIM undertake similar approaches, crucial differences exist between the two protocols. In particular, RIIM targets individual $cX$ gates afflicted by local depolarising noise. Being a noise-agnostic technique, by construction it is unable to correctly amplify (and therefore to suppress) noise processes that do not commute with the $cX$ gates~\cite{KWYMGTK21}. On the contrary, NOX targets entire cycles afflicted by a broad class of noise processes, including non-local and non-depolarising processes. By using Randomized Compiling in combination with CER, NOX can evaluate the the ability of Identity Insertion to correctly amplify the noise, and potentially use Append Errors to ensure a more precise amplification. This makes NOX more reliable than RIIM as well as more widely applicable.

\section{Our experiments}
\label{sec:experiment}


We begin this section by discussing our strategy for testing NOX and PEC. Next, we present the results of our experiments.

\subsection{Our testing strategy}
We conduct both numerical and experimental testing of PEC and NOX. The numerical testing allows us to evaluate the performance of our protocols in an ideal scenario in which the Pauli error rates of every cycle are known exactly and the assumptions A1 and A2 apply. In every simulation we model the cycles' noise based on the CER data collected in the corresponding experiment, and for simplicity we consider noiseless state preparation and measurements. On the other hand, with the experimental testing we investigate the performance of our protocols in a real-world setting, where the noise is known approximately but not exactly and slight deviations from A1 and A2 are to be expected|for example, non-Markovian errors have been previously observed on the chip that we use for our experiments \cite{KennethAndOthers21}.
 
We begin every experiment by characterizing the noise with CER. This typically takes around twenty minutes per noisy cycle. Next, we run the input circuit several times, with and without EM, in order to gather statistics for the final estimators. In addition to PEC and NOX we employ standard readout error mitigation (REM) protocols to mitigate measurement errors~\cite{BSKMG20}. The runtimes per repetition are usually within the hour. For example, for the three-qubit quantum phase estimation circuits (which contain $m=11$ noisy cycles), for $\sigma=2\%$ we require around one minute to compute an estimator for the unmitigated circuit, around ten minutes for the NOX estimator and around twenty minutes for the PEC estimator. We note that the errors afflicting idling qubits are the dominant type of error in our  device~(Fig.~\ref{fig:KNR} in Appendix~\ref{app:KNR}). As these errors commute with the noisy cycles in our circuits, to amplify the noise in NOX we use Identity Insertion with $\alpha=3$.

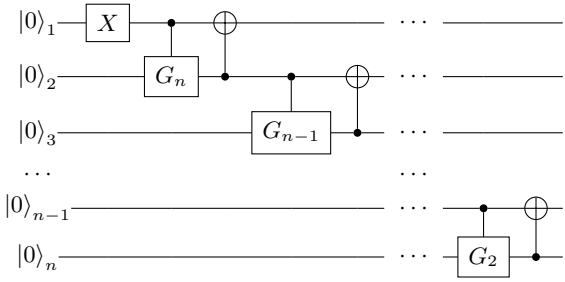
\begin{figure}[t]
\footnotesize
\begin{quantikz}[thin lines,
row sep=0.2cm,
column sep=0.2cm]
\ket{0}_1 & 
\gate{X} & 
\ctrl{1} &
\targ{} &
\qw &
\qw &
\qw &
\cdots & {}&
\qw &
\qw &
\qw \\
\ket{0}_2 & 
\qw & 
\gate{G_{n}} & 
\ctrl{-1} &
\ctrl{1} &
\targ{} &
\qw &
\cdots & {}&
\qw &
\qw &
\qw  \\
\ket{0}_3 & 
\qw & 
\qw &
\qw &
\gate{G_{n-1}} &
\ctrl{-1} &
\qw &
\cdots & {}&
\qw &
\qw &
\qw  \\
\cdots&&&&&&&\cdots&\\
\ket{0}_{n-1} &
\qw & 
\qw &
\qw &
\qw &
\qw &
\qw &
\cdots & {}&
\ctrl{1} &
\targ{} &
\qw \\
\ket{0}_n & 
\qw & 
\qw &
\qw &
\qw &
\qw &
\qw &
\cdots & {}&
\gate{G_{2}} &
\ctrl{-1} &
\qw 
\end{quantikz}
\caption{Circuit to generate an $n$-qubit W state \cite{CruzEtAl19}. The gates $G_t$ are defined in Eq.~\ref{eq:controlled-G}.}
\label{fig:w_circ}
\end{figure}

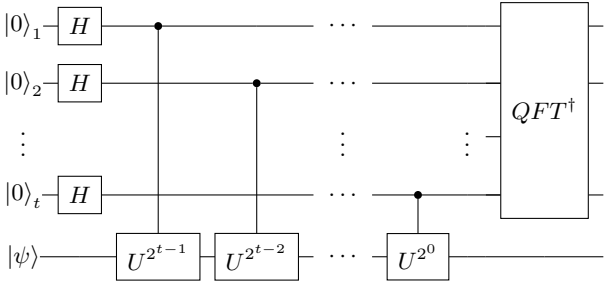
\begin{figure}[t]
\footnotesize
\begin{quantikz}[thin lines,
row sep=0.2cm,
column sep=0.2cm]
\ket{0}_1 & 
\gate{H} & 
\ctrl{4} &
\qw &
\qw &
\cdots & {}&
\qw &
\qw &
\qw & 
\gate[wires=4]{QFT^\dagger} &
\qw \\
\ket{0}_2 & 
\gate{H} & 
\qw & 
\ctrl{3} &
\qw &
\cdots &{}&
\qw &
\qw &
\qw &
\qw &
\qw \\
\vdots&{}&{}&{}&{}&\vdots&{}&{}&\vdots{}&{}&{}&{}\\
\ket{0}_{t} & 
\gate{H} & 
\qw &
\qw &
\qw &
\ldots & {}&
\ctrl{1} &
\qw &
\qw &
\qw &
\qw \\
\ket{\psi} &
\qw &
\gate{U^{2^{t-1}}} &
\gate{U^{2^{t-2}}} &
\qw &
\cdots & {}&
\gate{U^{2^{0}}} &
\qw &
\qw &
\qw &
\qw 
\end{quantikz}
\caption{Circuit to perform the QPE algorithm. In our tests we set $U=R_Z(\kappa)$, where $R_Z(\kappa)$ is defined in Eq.~\ref{eq:rz_QPE}, and run the algorithm for different values of $\kappa$.}
\label{fig:QPE_circ}
\end{figure}

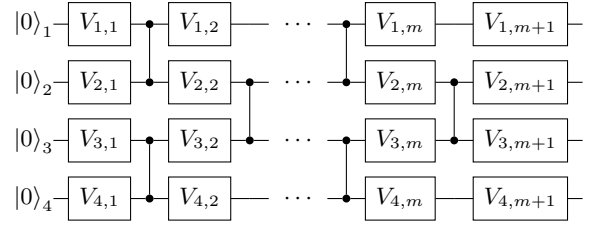
\begin{figure}[t]
\footnotesize
\begin{quantikz}[thin lines,
row sep=0.2cm,
column sep=0.2cm]
\ket{0}_1 & 
\gate{V_{1,1}} & 
\ctrl{1} &
\gate{V_{1,2}} &
\qw &
\qw &
\cdots & {} &\ctrl{1}&
\gate{V_{1,m}} &
\qw &
\gate{V_{1,m+1}} & 
\qw \\
\ket{0}_2 & 
\gate{V_{2,1}} & 
\control{} &
\gate{V_{2,2}} &
\ctrl{1}&
\qw &
\cdots & {} &\control{}&
\gate{V_{2,m}} &
\ctrl{1} &
\gate{V_{2,m+1}} & 
\qw \\
\ket{0}_3 & 
\gate{V_{3,1}} & 
\ctrl{1} &
\gate{V_{3,2}} &
\control{}&
\qw &
\cdots & {} &\ctrl{1}&
\gate{V_{3,m}} &
\control{} &
\gate{V_{3,m+1}} & 
\qw \\
\ket{0}_4 & 
\gate{V_{4,1}} & 
\control{} &
\gate{V_{4,2}} &
\qw &
\qw &
\cdots & {} &\control{}&
\gate{V_{4,m}} &
\qw &
\gate{V_{4,m+1}} & 
\qw \\
\end{quantikz}
\caption{Four-qubit pseudo-random circuits of the type implemented in our third test. Each gate $V_{i,j}$ is a random one-qubit gate.}
\label{fig:random_circuit}
\end{figure}

\subsection{W-state circuits}
\label{sec:w-state}
In our first test we implement our protocols on circuits that generate W states. W states are a special type of multipartite entangled states that play a central role in quantum communication, memories and networks~\cite{DP05,YBOATYO16}. An $n$-qubit W state can be written as an equal superposition of all the weight-one basis states, namely as
\begin{equation}
\label{eq:w_states}
    \ket{W_n}=\frac{\ket{0\ldots01}+\ket{0\ldots10}+\ket{1\ldots00}}{\sqrt{n}}\:.
\end{equation}
Fig.~\ref{fig:w_circ} shows a circuit to produce $n$-qubit W states in a linear architecture with nearest-neighbouring connectivity~\cite{CruzEtAl19}. This circuit contains $n-1$ controlled gates implementing the operation $|0\rangle\langle0|\otimes I+|1\rangle\langle1|\otimes G_t$, where
\begin{equation}
    \label{eq:controlled-G}
    G_t=
    \begin{pmatrix}
    \sqrt{\frac{1}{t}} & - \sqrt{1-\frac{1}{t}}\\\\
    \sqrt{1-\frac{1}{t}} & \sqrt{\frac{1}{t}}
    \end{pmatrix}\:.
\end{equation}
Each one of these gates is followed by a $cX$ gate. After recompiling the entangling gates into our native set ($cZ$ gates between nearest neighbours), the resulting circuit contains $3(n-1)$  noisy cycles, each one comprising one $cZ$ gate and two identity gates.

In our tests we generate W states with $n=2,\:3$ and 4 qubits. By measuring these states in the computational basis we estimate the probability associated with each one of the possible outputs. Formally, we compute the quantity $p_\textrm{est}(\bar{s})=$Tr$\big[O_{\bar{s}}|W_n\rangle\langle W_n|\big]$ for every projector $O_{\bar{s}}\in\{|{\bar{s}}\rangle\langle{\bar{s}}|\::\:{\bar{s}}\in(0,1)^{\otimes n}\}$. Fig.~\ref{fig:w_states} shows the estimates for the most frequent outputs obtained in the simulations (Fig.~\ref{fig:w1}) and in the experiment (Fig.~\ref{fig:w2}). As it can be seen, all the estimates returned by NOX and PEC concentrate around (or very close to) their ideal value, 
whereas the unmitigated estimates are generally inaccurate. To better quantify the improvement provided by our EM protocols we calculate the variation distance
\begin{equation}
\label{eq:def_VD}
    \textrm{VD}:=\frac{1}{2}\sum_{\bar{s}}\big|p_\textup{id}(\bar{s})-p_\textup{est}(\bar{s})\big|
\end{equation}
 between $\{p_\textrm{est}(\bar{s})\}$ and the ideal probability distribution of the outputs $\{p_\textrm{id}(\bar{s})\}$. As shown in Fig.s \ref{fig:w3} and \ref{fig:w4}, the variation distances of the mitigated outputs are significantly smaller than those of the unmitigated outputs, with average improvements between $47\%$ and $66\%$ for the experimental outputs~(Fig.~\ref{fig:improvements}).

\subsection{Quantum Phase Estimation algorithm}
In our second test we run a quantum phase estimation~(QPE) experiment. QPE is an important primitive required by many quantum algorithms, including Shor's factoring algorithm \cite{NC00}. Given a gate $U$, a ``target'' state $\ket{\psi}$ and a (potentially unknown) number $\kappa\in[0,1)$ such that $U\ket{\psi}=\textup{exp}(2i\kappa\pi)\ket{\psi}$, the task of QPE is to produce an estimate $\widehat{\kappa}$ of $\kappa$. To do so, QPE requires initialising $t\geq1$ ancillae, entangling each of them with $\ket{\psi}$ and eventually performing the inverse Quantum Fourier Transform~(QFT), as shown in Fig.~\ref{fig:QPE_circ}. By measuring the ancillae and post-processing the outputs, QPE returns an estimate $\widehat{\kappa}\in\{p/2^t\::\:p\in\{0,\ldots,2^t-1\}\}$ such that $|\kappa-\widehat{\kappa}|\leq2^{-t}$ with high probability. 

Setting $t=2$ and $\ket{\psi}=\ket{1}$, we estimate the parameter $\kappa$ for a series of gates that perform rotations of the type
\begin{equation}
\label{eq:rz_QPE}
R_Z(\kappa)=\textup{diag}(1, \textrm{exp}(2i\kappa\pi))\:.
\end{equation}
Decomposed into our native gateset, our QPE circuits contain $n=3$ qubits and $m=11$  noisy cycles. As with the W-state circuits, we measure all the qubits (target and ancillae) in the computational basis and reconstruct the probability distribution of the outputs. Fig.~\ref{fig:QPE} shows the variation distances between ideal and estimated probability distributions of the outputs obtained in the simulations (Fig.~\ref{fig:QPE1}) and in the experiments (Fig.~\ref{fig:QPE2}). In both cases the mitigated outputs are significantly more accurate than the unmitigated ones, with average improvements between $42\%$ and $86\%$ in variation distance for the experimental outputs (Fig.~\ref{fig:improvements}).

The  better  accuracy  of  the  outputs  under  EM  naturally improves the precision of the QPE algorithm. To see this, by post-processing the estimated probability distributions of the outputs of the ancillae we calculate the probability $q_\textrm{est}(\widehat{\kappa}|\kappa)$ that the QPE algorithm returns $\widehat{\kappa}\in\{0.00,0.25,0.50,0.75\}$ when $\kappa$ is the parameter being estimated. Fig.~\ref{fig:QPE_kappas_3qb} shows the probabilities $q_\textrm{est}(\widehat{\kappa}|\kappa)$ calculated in the various experiments (solid bars), along with the ideal probabilities $q_\textrm{id}(\widehat{\kappa}|\kappa)$ calculated with a noiseless simulation (striped bars). The probabilities $q_\textrm{est}(\widehat{\kappa}|\kappa)$ obtained in the experiments with PEC and NOX are generally closer to the ideal ones than those obtained in the experiments without EM. To quantify the improvement, we calculate the variation distances
\begin{equation}
\label{eq:vd_QPE}
    \textrm{VD}_\textrm{QPE}^{(\kappa)}:=\frac{1}{2}\sum_{\widehat{\kappa}}\big|q_\textrm{id}(\widehat{\kappa}|\kappa)-q_\textrm{est}(\widehat{\kappa}|\kappa)\big|
\end{equation}
between the ideal and estimated probability distributions of the outcomes of QPE for the values of $\kappa$ chosen in our experiments. As shown in Table~\ref{tab:vd_QPE}, NOX and PEC drastically improve the precision of the QPE algorithm in all the cases considered.

We repeat the above experiment with $t=3$ ancillae, setting $\ket{\psi}$ and $\kappa=0.5$. Decomposed into our native gateset, the resulting QPE circuit contains $n=4$ qubits and $m=25$ noisy cycles. As opposed to the experiments with $t=2$ ancillae, both NOX and PEC return less accurate outputs than the unmitigated circuit and visibly decrease the precision of the QPE algorithm (Fig.~\ref{fig:QPE_kappas_4qb}). We attribute this unsuccessful result to a series of noise processes (unmodeled non-Markovian errors~\cite{KennethAndOthers21}, drift, errors due to an inaccurate amplification of the noise, etc.) that are not mitigated by our protocols. These unmitigated noise processes accumulate along the circuit, resulting in a bias that grows linearly in $m$ and that becomes non-negligible in deep circuits. This failed test shows us that when the noise processes that are not encompassed by our assumption become dominant, they have the potential to disrupt the performance of our EM protocols.

\begin{table}[t]
\footnotesize
    \centering
    \setlength{\extrarowheight}{4pt}
    \begin{tabular}{r||g|y|g}
    $\kappa$ & \textbf{NOX+REM} & \textbf{PEC+REM} & \:\textbf{REM}\: \\
    \hline
    \hline
   &&\\[-3.7ex]
   \textbf{0.2} & $0.6\%\textrm{ } (\uparrow 96.2\%)$ & $3.1\%\textrm{ } (\uparrow 80.6\%)$ & $16.1\%$ \\ 
   &&\\[-3.7ex]
   \textbf{0.4} & $3.8\%\textrm{ } (\uparrow 80.2\%)$ & $2.8\%\textrm{ } (\uparrow 85.6\%)$ & $19.2\%$ \\
   &&\\[-3.7ex]
    \textbf{0.6} & $2.8\%\textrm{ } (\uparrow 74.1\%)$ & $3.3\%\textrm{ } (\uparrow 69.4\%)$ & $10.8\%$ \\
   &&\\[-3.7ex]
    \textbf{0.8} & $1.3\%\textrm{ } (\uparrow 74.0\%)$ & $3.7\%\textrm{ } (\uparrow 26.0\%)$ & $5.0\%$ \\
\end{tabular}
    \caption{Values of $ \textrm{VD}_\textrm{QPE}^{(\kappa)}$ (Eq.~\ref{eq:vd_QPE}) calculated for the experiments with NOX+REM, PEC+REM and REM. The values in parenthesis correspond to the improvement over the corresponding REM value.}
    \label{tab:vd_QPE}
\end{table}

\subsection{Pseudo-random circuits}
In our third test we target pseudo-random circuits of varying depth of the type shown in Fig.~\ref{fig:random_circuit}. These circuits alternate between cycles containing either one or two $cZ$ gates and cycles containing random one-qubit gates. Fig.~\ref{fig:random} shows the variation distances between ideal and estimated probability distributions of the outputs obtained numerically~(Fig.~\ref{fig:random1}) and experimentally~(Fig.~\ref{fig:random2}). As in our previous tests, the mitigated outputs are visibly more accurate than the unmitigated ones, with average improvements between $32\%$ and $56\%$ in variation distance for the experimental outputs (Fig.~\ref{fig:improvements}).

\subsection{Relation between input the parameter $\sigma$ and the standard deviation of the estimators.}
In addition to suppressing the bias of the final estimators, our protocols provide guarantees about their statistical fluctuations. In particular, choosing a specific value for the input $\sigma$ guarantees $O(\sigma)$ standard deviation for every estimator, at the cost of running a number $N=O(\sigma^{-2})$ of circuits. To verify the relation between the input $\sigma$ and the standard deviation, we test numerically the performance of NOX on a two-qubit W-state circuit for different values of $\sigma$. As shown in Fig.~\ref{fig:sigma1}, smaller values of $\sigma$ lead to estimators that are statistically more accurate, which confirms the expected relation between $\sigma$ and the standard deviation under ideal experimental conditions. 

In a real-world setting, a number of uncontrollable factors (such as drift in the noise afflicting the  device in use) may inevitably prompt fluctuations in the estimators, limiting our ability to attain the desired statistical accuracy. To see how this may affect our protocols, we repeat our two-qubit test experimentally (Fig.~\ref{fig:sigma2}). We find that for $\sigma\gtrsim2\%$ the standard deviation of the results decreases with $\sigma$ as expected, whereas for $\sigma\lesssim2\%$ the standard deviation remains approximately constant. Due to the inherent fluctuations of the noise afflicting the device in use, implementing NOX with $\sigma<2\%$ requires running more circuits than with $\sigma=2\%$, but it does not improve the statistical accuracy of the estimators. In other words, our EM protocols cannot provide performance guarantees below the noise floor of the device. Overall, the results shown in Fig.~\ref{fig:sigma2} highlight the importance of the assumption A1 in the context of EM and call for methods to suppress drifts.

\section{Conclusions}
While fault-tolerance remains a long-term goal, understanding how to improve the performance of the existing noisy quantum computers is of timely importance. By leveraging cutting-edge protocols for noise reconstruction, we have developed PEC and NOX and experimentally tested their effectiveness and practicality on a four-qubit superconducting chip. The results of our tests demonstrate that both of our protocols can significantly enhance the performance of the noisy quantum circuits implemented on existing hardware platforms.

The previous EM protocols based on noise reconstruction are centered around GST \cite{TBG17,EBL18}. Since GST is inefficient in the number of qubits, these protocols have been tested on circuits containing up to two qubits \cite{SongEtAl19,ZhangEtAL20}. Implementation on larger circuits required enhancing the noise-reconstruction process with machine learning tools, at the price of increased complexity and runtime~\cite{SQCBY20}. On the contrary, being robust to all the main noise processes that naturally occur in multi-qubit systems, our protocols provide the tools to increase the performance of platforms with an arbitrary number of qubits, provided that they suffer moderate levels of noise.

Going forward, it is important to study how EM can help bridge the gap between today's noisy devices and tomorrow's fault-tolerant quantum computers (FTQC). Recent works showed how EM can reduce various types of logical errors in FTQCs, such as errors due to insufficient code distances \cite{PSBGT21} or imperfect magic-state distillation \cite{PSBGT21,SEFT22}. Due to their ability to suppress multi-qubit errors, we anticipate that our EM protocols may be helpful to suppress multi-qubit physical errors that have a higher weight than the code corrects, and consequently to reduce the errors at the logical level. We leave this point open for future works.\\

\noindent\textbf{Note added.} While editing the final version of this manuscript, we became aware of related work that also employs efficient methods for noise reconstruction to enhance error cancellation \cite{vanDenBergAndOthers22}. Our protocols have been developed independently and around the same time as that in Ref.~\cite{vanDenBergAndOthers22}.\\

\noindent\textbf{Acknowledgments.} This research was undertaken thanks in part to funding from the Canada First Research Excellence Fund, the Government of Ontario, and the Government of Canada through NSERC. This work was supported by the U.S.~Department of Energy, Office of Science, Office of Advanced Scientific Computing Research Quantum Testbed Program under Contract No.~DE-AC02-05CH11231. 

S.F.~and A.H.~contributed equally to this work. S.F.~and J.J.W.~conceptualized the work. A.H., J.-L.V., and A.M.~conducted the experiments. R.N., D.I.S., and I.S.~supervised all experimental work. S.F., A.H.~and J.J.W.~wrote the manuscript, with input from all coauthors. S.F.~thanks J.~Skanes-Norman for useful discussions.

\section{Appendix}
\subsection{Further details about noise reconstruction.}~\label{app:KNR}
In this Appendix we discuss the noise-reconstruction data obtained in our experiments. We begin with the CER data, then we discuss the readout calibration (RCAL) data.\\

\noindent\textbf{CER data.} Fig.~\ref{fig:KNR} shows the data obtained in two different implementations of CER. We note that in both figures, the errors afflicting the idling qubits dominate the probability distributions. This is due to the nature of the $cZ$ gates employed in this work, which utilized off-resonant drives to implement tunable $ZZ$ interaction between two fixed-frequency transmon qubits \cite{MitchellAndOthers21}. These off-resonant drives can induce phase errors on spectator qubits if the driving tones are far off-resonant, or induce partial Rabi driving on spectator qubits (i.e.~$X$- or $Y$-type errors) if the driving tones are near-resonant. The $cZ$ gates also contain signals (with equal amplitude but opposite phase) which attempt to null any conditional errors on the neighboring spectator qubits, but if the nulling of these crosstalk terms is imperfect (as we see in Fig.~\ref{fig:KNR}), then errors on the idling qubit can dominate the CER results. 

We do not observe significant fluctuations of the cycles' error rates over periods of several hours. Thus, to minimize the runtime, we avoid taking new CER data in between different repetitions of the same circuit. See the Supplemental Material for Ref.~\cite{Hashim20} for further details about CER and errors on this device.\\

\noindent\textbf{RCAL data.} To obtain the RCAL data, we implement circuits of two different types. Firstly, we implement circuits with an identity on every qubit. Secondly, circuits with a Pauli-$X$ on every qubit. We use the relative frequency of the outputs 0 and 1 on every qubit to estimate state-dependent measurement errors. Fig. \ref{fig:rcal_fit} shows RCAL data taken on September 7, 2021. As it can be seen, the probability that an outcome $0$ is flipped to $1$ is below $1\%$ for every qubit, while the probability that an outcome 1 is flipped to a 0 is around $2\%$ on average (Fig.~\ref{fig:rcal_fit}).

We collect the RCAL data at the beginning of each experiment, and we avoid taking new RCAL data while running the experiment in order to minimise the runtime. We use the RCAL data to mitigate the readout noise via REM~\cite{BSKMG20}. By applying REM we observe a visible improvement of the results for W-state (Fig.s \ref{fig:w_states_app}). On the contrary, when we apply REM to circuits with a higher two-qubit gate count (Fig.s \ref{fig:qpe_app} and \ref{fig:random_app}), we obtain outputs that are equal to the unmitigated ones within error bars. This remarks the fact that mitigating both cycles' and readout errors can be significantly more beneficial than mitigating readout errors alone, especially for circuits that contain a large number of two-qubit gates.

\onecolumngrid

\begin{figure}[H]
     \centering
      \subfloat[][\small  Numerical testing, measured estimators for the most frequent outputs.]{\includegraphics[clip,width=0.95\columnwidth]{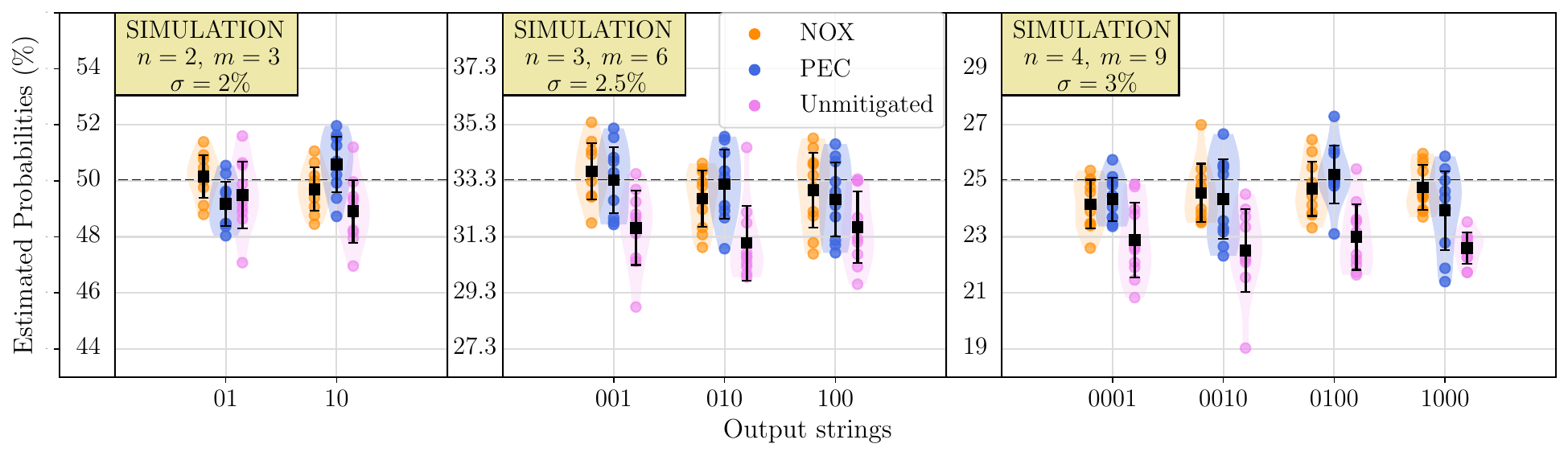}\label{fig:w1}}
     \qquad
     \qquad
      \subfloat[][\small  Experimental testing, measured estimators for the most frequent outputs.]{\includegraphics[clip,width=0.95\columnwidth]{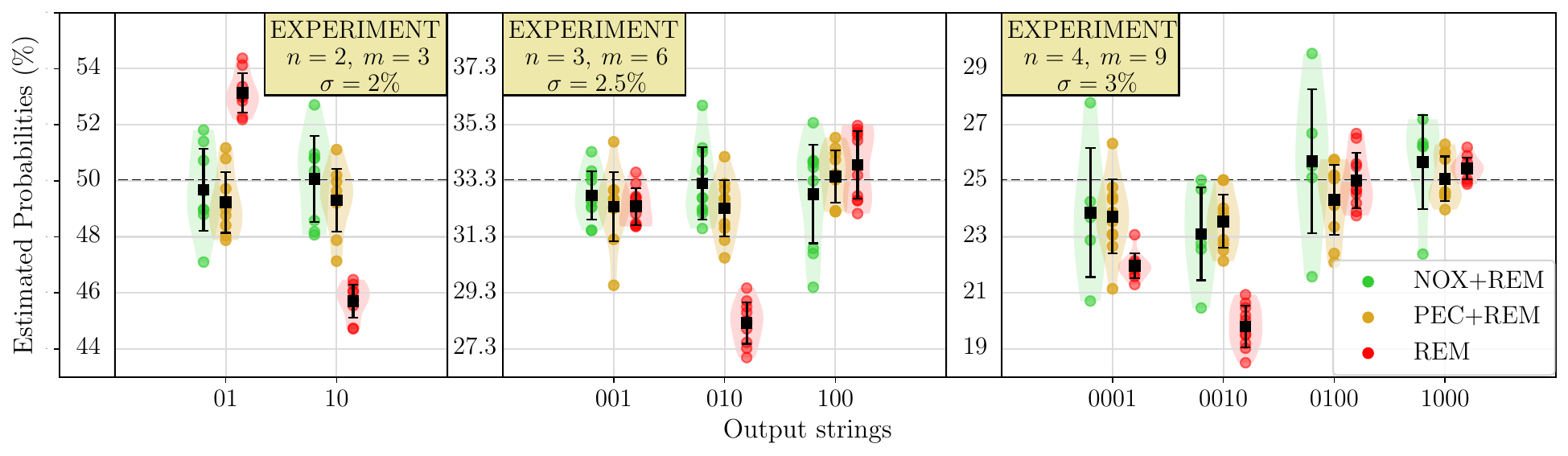}\label{fig:w2}}
     
      \subfloat[][\small Numerical testing, variation distances.]{\includegraphics[clip,width=0.45\columnwidth]{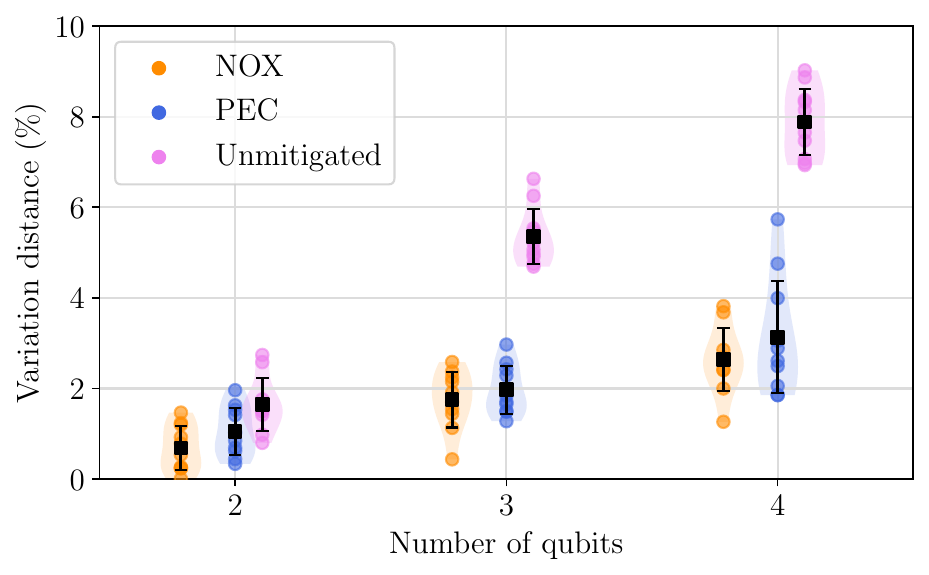}\label{fig:w3}}
      \subfloat[][\small Experimental testing, variation distances.]{\includegraphics[clip,width=0.45\columnwidth]{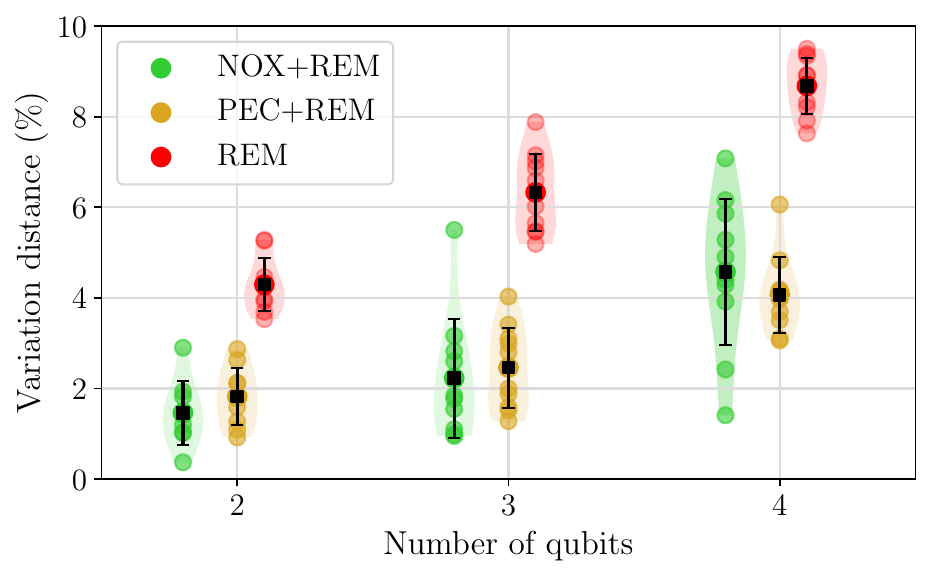}\label{fig:w4}}
     \caption{Summary of the results obtained for the W-state circuits. Figs.~(a) and (b) show the estimated probabilities for the most frequent outputs obtained in the simulations and experiments, respectively.
     Figs.~(c) and (d) show the variation distances between ideal and estimated probability distributions. In every figure the dots correspond to the actual data, the squares represent their means, the bars their standard deviations, and the dashed lines their ideal values. Note that for every $n$, the numerical estimates for the unmitigated circuit in (a) concentrate around similar values, which are close to the ideal value for $n=2$ and below the ideal value for $n=3,4$. Instead, the corresponding experimental estimates in~(b) are subject to larger fluctuations, and in some cases they are above the ideal value (see, for example, the estimates for the output 01). This is due to the presence of coherent errors in the experimental implementation of the unmitigated circuit. These coherent errors are tailored by Randomized Compiling and reported by CER in the form of stochastic Pauli errors. As a result, our simulations (which model the noise based on CER data) do not capture the full impact of these coherent errors.}
     \label{fig:w_states}
\end{figure}

\begin{figure}[H]
     \centering
      \subfloat[][\small Numerical testing with $t=2$ ancillae.]{\includegraphics[clip,width=0.45\columnwidth]{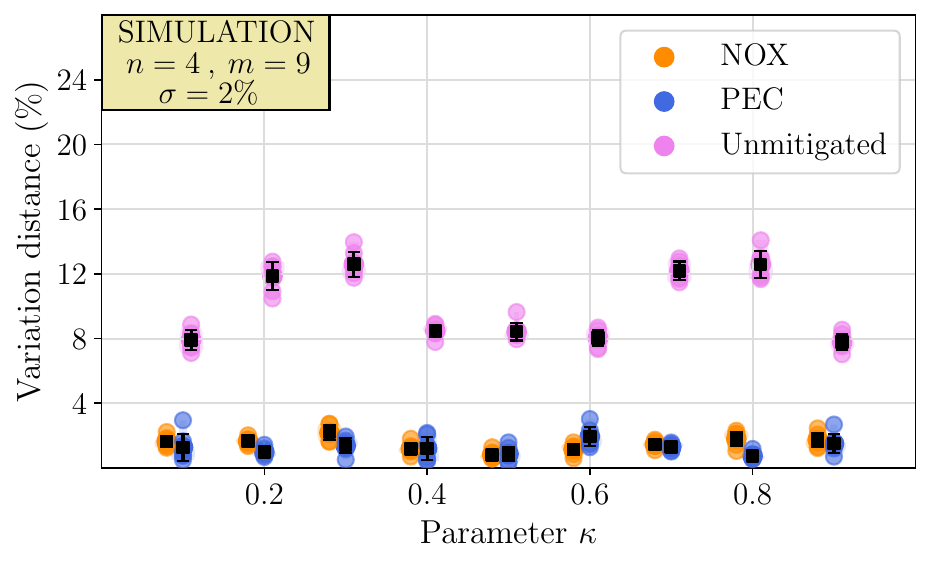}\label{fig:QPE1}}
     \qquad
     \qquad
      \subfloat[][\small Experimental testing with $t=2$ ancillae.]{\includegraphics[clip,width=0.45\columnwidth]{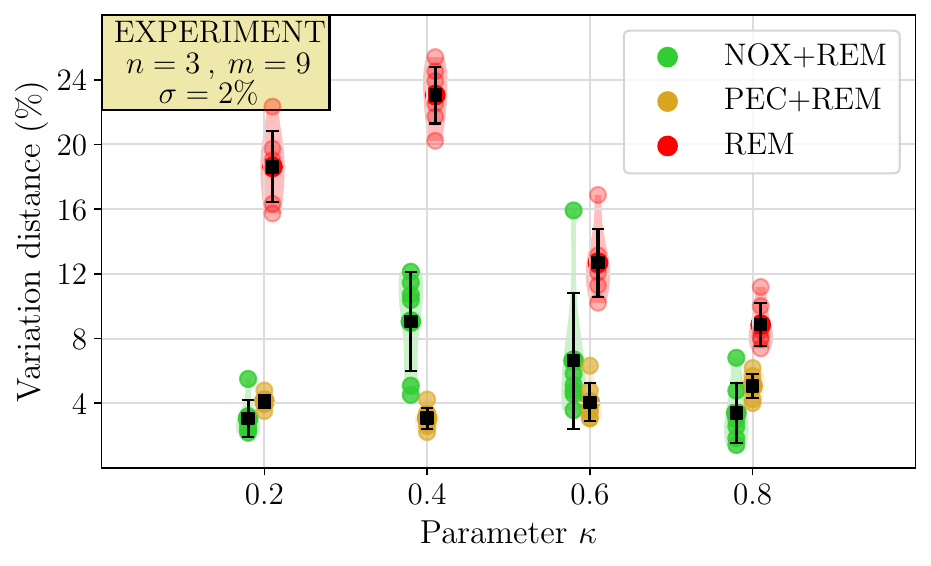}\label{fig:QPE2}}
      
     \subfloat[][\small Estimated probability of the parameters $\kappa_\textrm{est}$ returned by the QPE algorithm in the various experiments with $t=2$ ancillae.]{\includegraphics[clip,width=1\columnwidth]{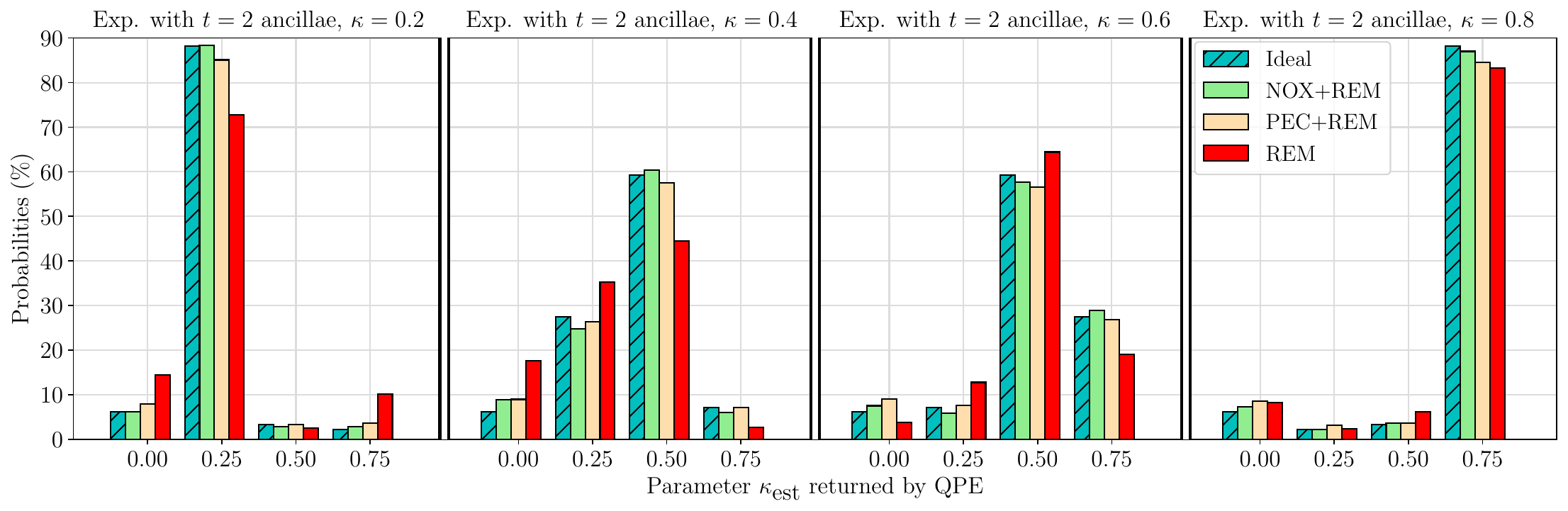}\label{fig:QPE_kappas_3qb}}
      
     \subfloat[][\small Estimated probability of the parameters $\kappa_\textrm{est}$ returned by the QPE algorithm with $t=3$ ancillae and $\kappa=0.5$.]{\includegraphics[clip,width=1\columnwidth]{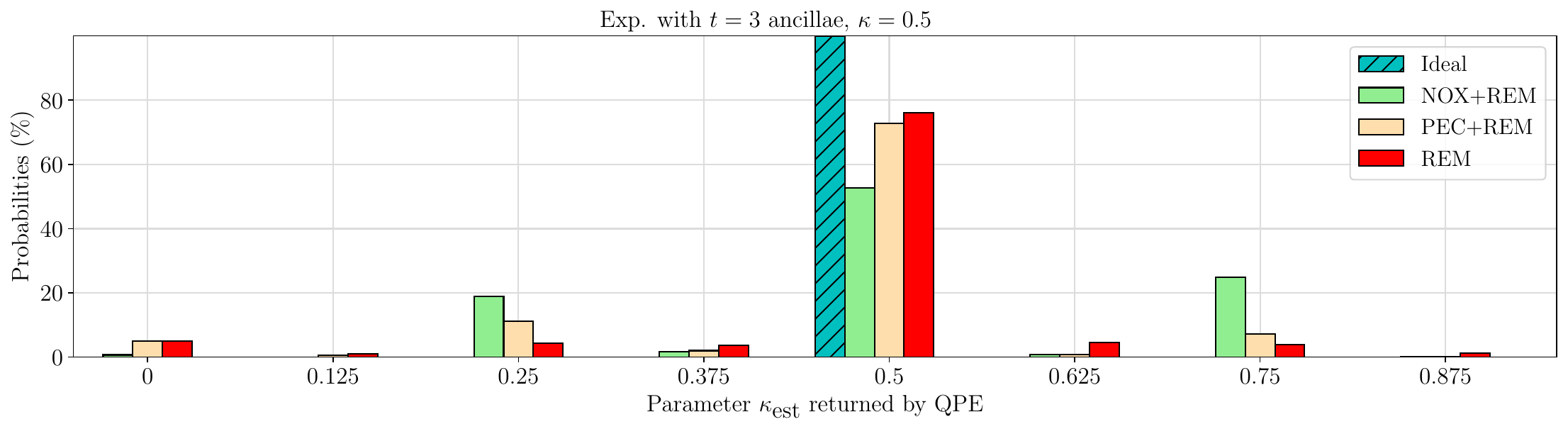}\label{fig:QPE_kappas_4qb}}
     \caption{\small Summary of the results for the QPE circuits. Figs.~(a) and (b) show the variation distances between ideal and estimated probability distributions of the outputs obtained in the simulations and in the experiments for our QPE experiments with $t=2$ ancillae. The dots correspond to the actual data, the squares represent their means and the bars their standard deviations. Fig.~(c) shows the estimated probabilities that the QPE algorithm with $t=2$ ancillae returns a given parameter $\widehat{\kappa}_\textrm{est}$|the striped bars are calculated with noiseless simulations, the solid bars are calculated by averaging over the experimental outputs. Fig.~(d) shows the estimated probabilities that the QPE algorithm with $t=3$ returns  a given parameter $\widehat{\kappa}_\textrm{est}$.}
     \label{fig:QPE}
\end{figure}

\begin{figure}[H]
     \centering
      \subfloat[][\small Numerical testing with random circuits.]{\includegraphics[clip,width=0.45\columnwidth]{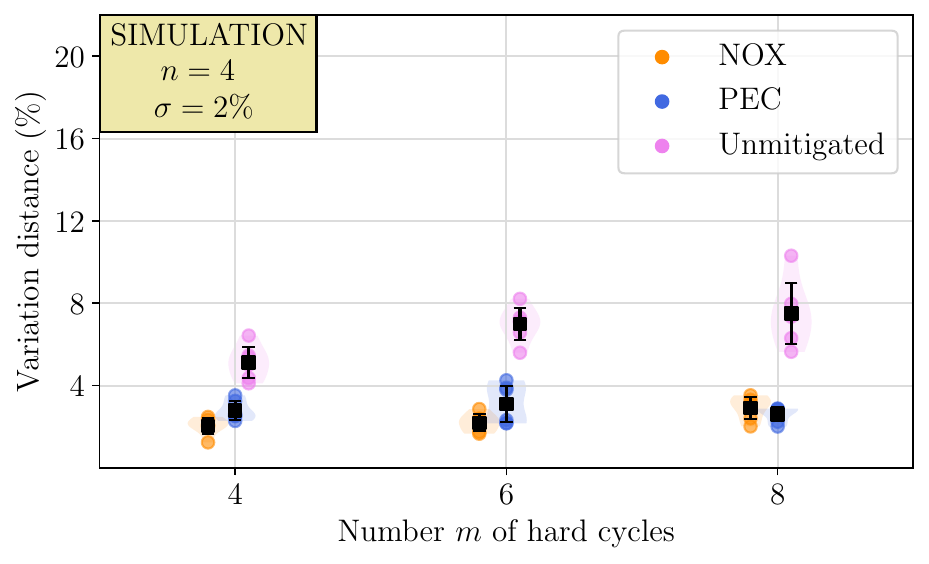}\label{fig:random1}}
     \qquad
     \qquad
      \subfloat[][\small Experimental testing with random circuits.]{\includegraphics[clip,width=0.45\columnwidth]{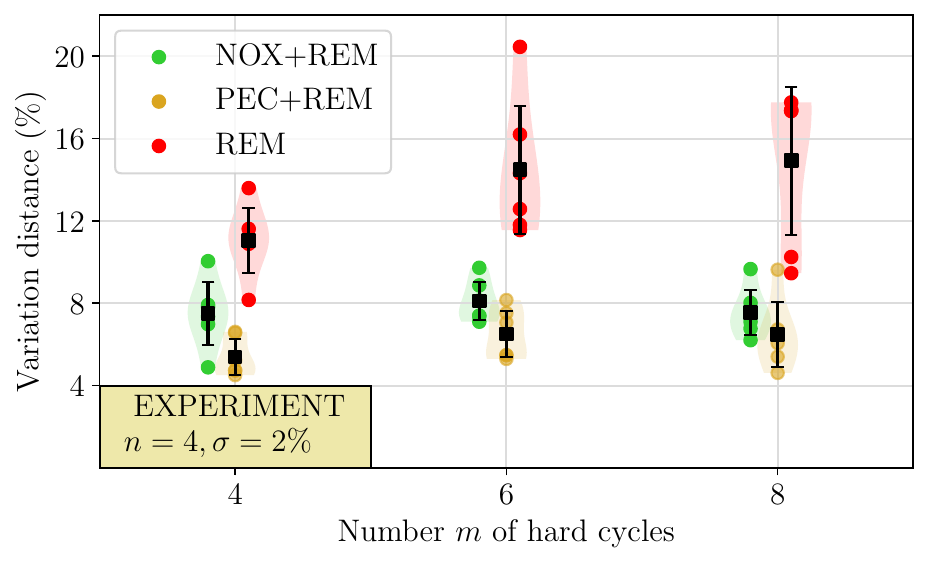}\label{fig:random2}}
     \caption{\small Summary of the results for the pseudo-random circuits. Figs.~(a) and (b) show the variation distances between the ideal and estimated probability distributions of the outputs obtained in the simulations and experiments, respectively. The dots correspond to the actual data, the squares represent their means and the bars their standard deviations.}
     \label{fig:random}
\end{figure}

\begin{figure}[H]
     \centering
      \subfloat[][\small Numerical testing.]{\includegraphics[clip,width=0.45\columnwidth]{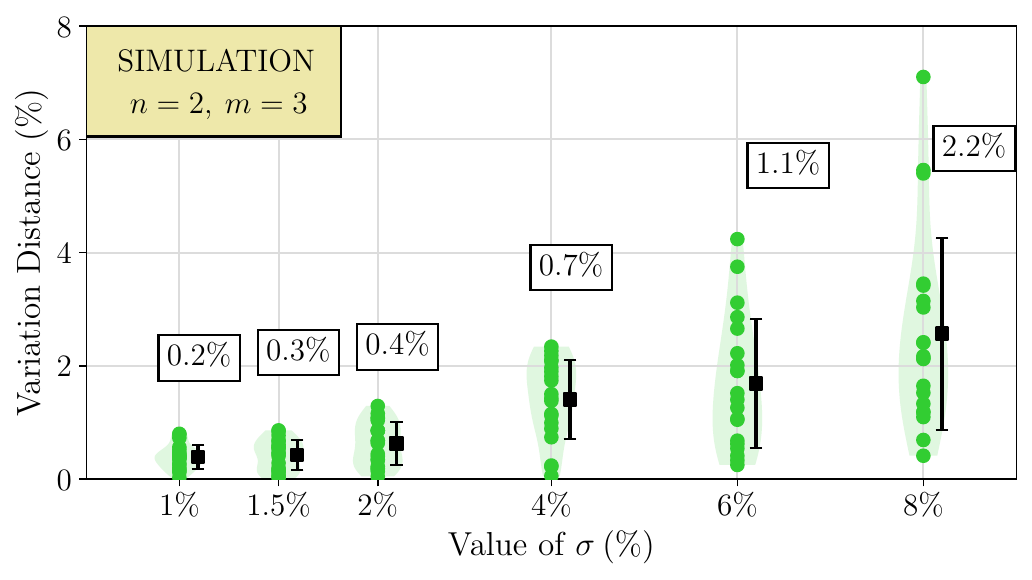}\label{fig:sigma1}}
     \qquad
     \qquad
      \subfloat[][\small Experimental testing.]{\includegraphics[clip,width=0.45\columnwidth]{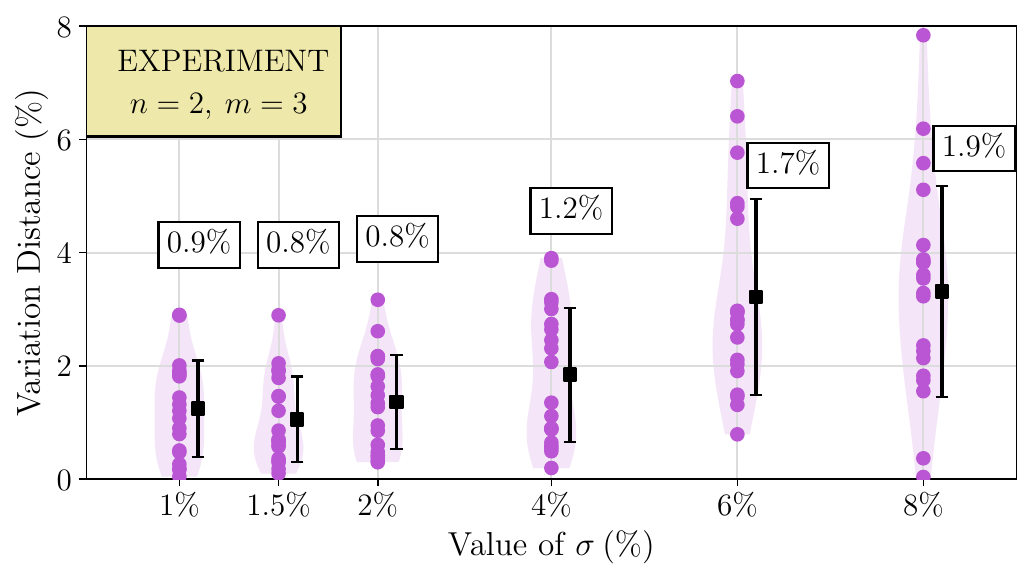}\label{fig:sigma2}}
     \caption{\small Summary of the results obtained by applying NOX on a two-qubit W-state circuit at different values of $\sigma$. Figs.~(a) and (b) show the variation distances between ideal and estimated probability distributions of the outputs obtained numerically and experimentally, respectively. The dots correspond to the actual data, the squares represent their means and the bars their standard deviations (which are reported in detail in the white boxes).}
     \label{fig:sigma}
\end{figure}

\begin{figure}[H]
    \centering
    \includegraphics[width=3.cm]{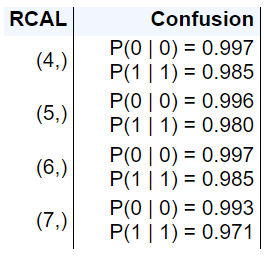}
    \caption{Readout calibration (RCAL) estimates obtained on September 7, 2021 by running around 10,000 calibration circuits. The l.h.s. column contains qubit labels, the r.h.s. column contains estimates of the probabilities of no error|specifically, for $i\in\{0,1\}$, P$(i|i)$ represents the conditional probability that a measurement returns $i$, given that $i$ is expected outcome.}
    \label{fig:rcal_fit}
\end{figure}

\newpage
\begin{figure}[H]
     \centering
      \subfloat[][CER data taken on June 16, 2021.]{\includegraphics[clip,width=0.65\columnwidth]{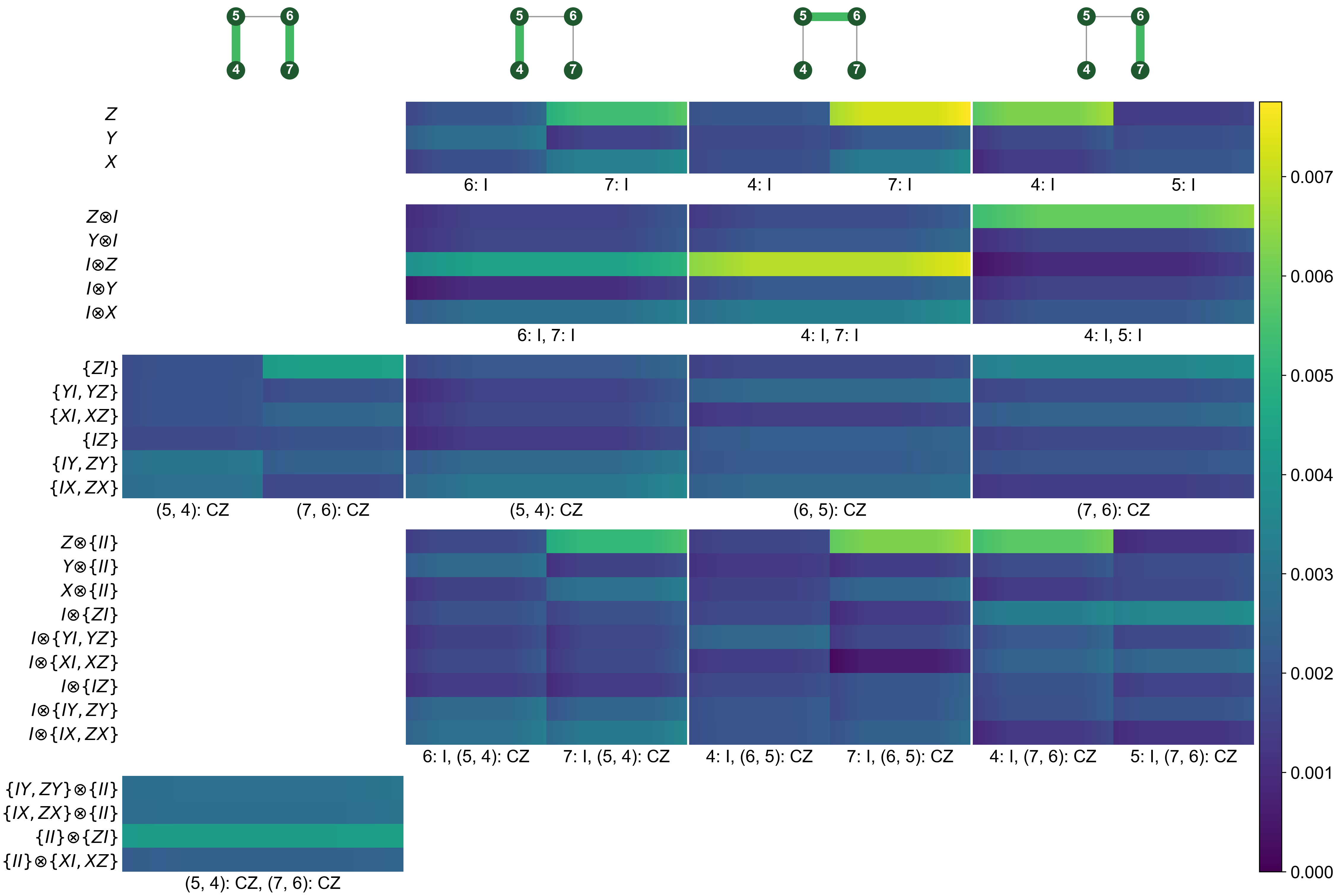}\label{fig:KNR_june}}
      
      \qquad
      \qquad
      
      \subfloat[][CER data taken on September 7, 2021.]{\includegraphics[clip,width=0.65\columnwidth]{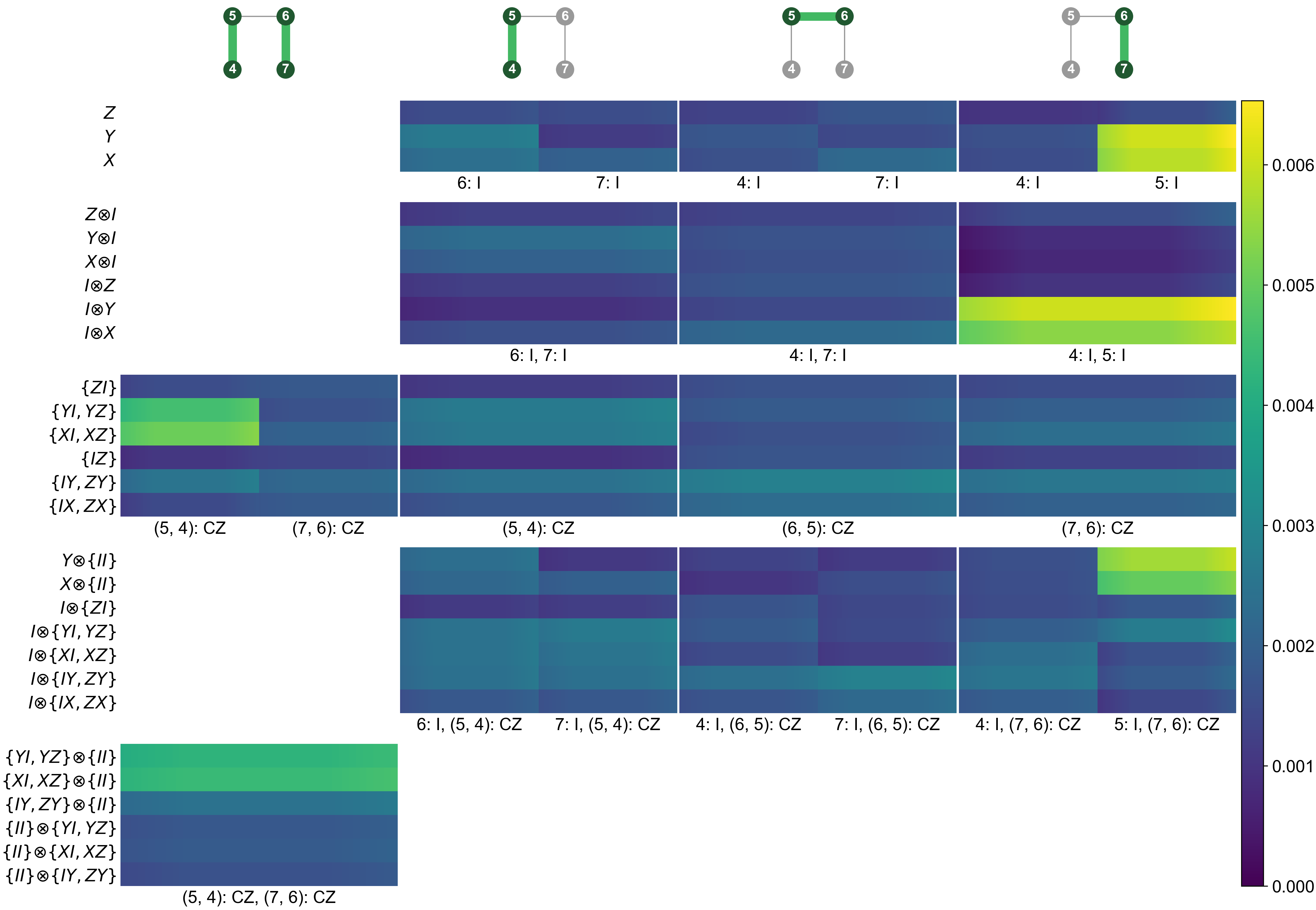}\label{fig:KNR_sept}}
     \caption{\small A map of the Pauli error rates (Eq.~\ref{eq:pauli_error_rates}) for different four-qubit cycles performed by the chip in Fig.~\ref{fig:chip}, each containing either one or two $cZ$ gates. The errors on the l.h.s. correspond to the Pauli errors afflicting the cycles, the colormap indicates the estimated probabilities for the Pauli errors in the plot, and the gradient across each cell defines the 90\% confidence interval of each estimate. All the errors with negligible probabilities are truncated and are not displayed. The first row of subplots shows the weight-one errors acting on the idling qubits. The second row shows the weight-one and weight-two errors on the idling qubits. The third row shows the weight-one and weight-two errors on the non-idling qubits. The fourth and fifth rows show correlated errors afflicting more than two qubits.}
     \label{fig:KNR}
\end{figure}

\newpage
\begin{figure}[!t]
     \centering
      \subfloat[][W-state circuits.]{\includegraphics[clip,width=8cm]{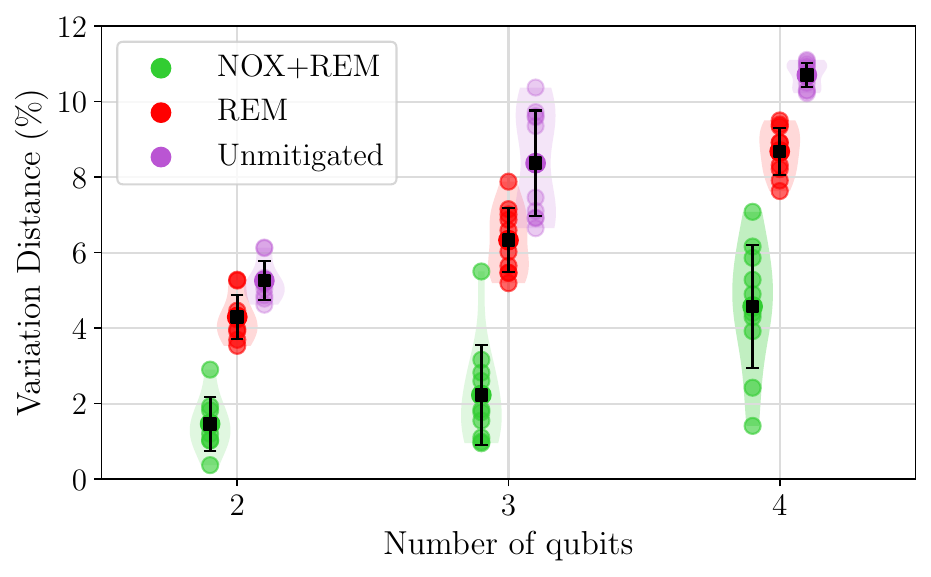}\label{fig:w_states_app}}
      \qquad
      \subfloat[][Three-qubit QPE circuits.]{\includegraphics[clip,width=8cm]{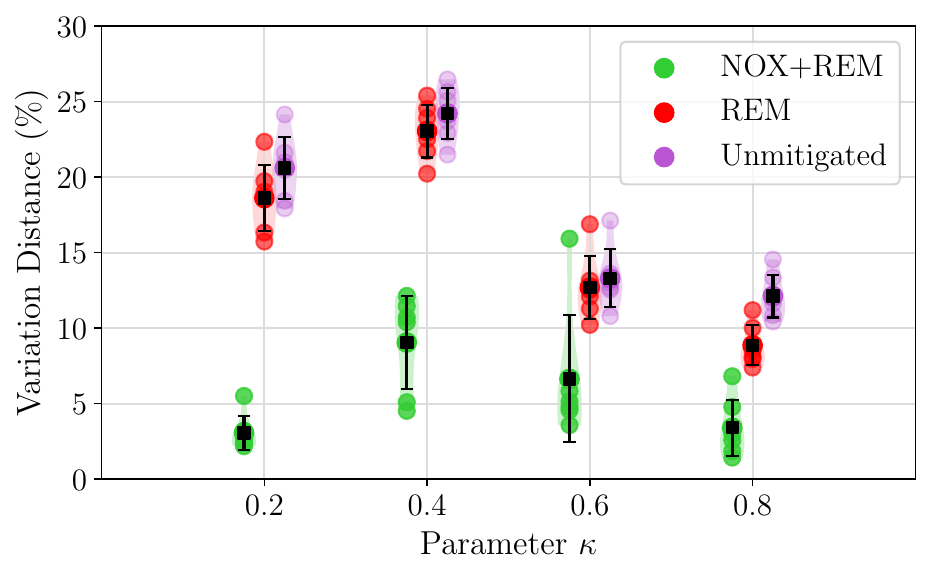}\label{fig:qpe_app}}
      
      \qquad
      \qquad
      
      \subfloat[][Random circuits.]{\includegraphics[clip,width=8cm]{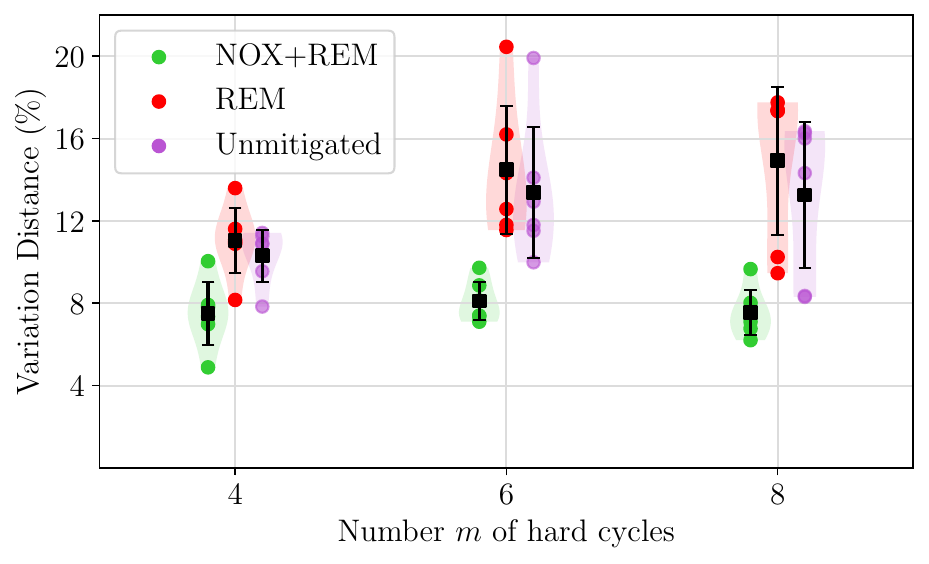}\label{fig:random_app}}
    \caption{Summary of the results obtained in the various experiments by applying NOX+REM, by applying REM alone and by applying no mitigation. As it can be seen, REM improved the outputs of our W-state (Fig. a), but it did not significantly improve the outputs of our QPE and random circuits (Fig.s b, c).}
    \label{fig:results_app}
\end{figure}

\twocolumngrid

\bibliographystyle{plainnat}
\bibliography{biblio}

\onecolumngrid
\newpage

\begin{center}
\textbf{\large Supplementary Material}\\
\end{center}

\begin{center}
\textbf{I. Formal description of PEC and proof of theorem \ref{th:PEC}}
\end{center}

In this section we provide a formal description of our PEC protocol, then we show a proof of theorem~\ref{th:PEC}. The box below shows the algorithm implemented by PEC:

\vspace{0.6cm}

\begin{center}
\fbox{
    \parbox{0.8\textwidth}{
\begin{small}
\textbf{Box \hypertarget{PEC}{1}.} Pauli Error Cancellation (PEC)\\
\noindent\rule{0.8\textwidth}{0.4pt}\\
\textbf{Inputs.} \\
An $n$-qubit quantum circuit $\C=\E_{m+1}\h_{m}\E_{m}\cdots\h_1\E_1$, the Pauli error rates $\{\epsilon_l^{(\h_j)}\}$ of the  noisy cycles $\h_j$, an $n$-qubit state $\rho_\textup{in}$, an operator $O$ such that the spectral norm $||O||_\infty\sim 1$ and a number $\sigma\in(0,1)$.
\sbline
\textbf{Routine.}\\
\textbf{1.} Calculate $C_{\textup{tot}}$ (Eq.~\ref{eq:cost}) and $N=(C_\textup{tot}/\sigma)^2$.
\sbline
\textbf{2. }For $k\in\{1,\ldots,N\}$:
\begin{itemize}[leftmargin=1.2cm]
    \item[\textbf{2.1 }]For $j\in\{1,\ldots,m\}$:
    \item[] Choose $\p^{(k)}_{j}\in\{\I,\X,\Y,\Z\}^{\otimes n}$ at random with probability $\epsilon_{l_j}^{(\h_j)}$. Next, initialise the circuit $\C^{\textup{(PEC)}}(\p^{(k)}_1,\ldots,\p^{(k)}_m)$ as in Eq.~\ref{eq:circuit_PEC}.
    \item[\textbf{2.2}] If an even number of Pauli operators $\p^{(k)}_1,\ldots,\p^{(k)}_m$ is equal to the identity, initialise $s_k=1$. Otherwise, initialise $s_k=-1$.
\end{itemize}

\textbf{3. }For $k\in\{1,\ldots,N\}$:
\begin{itemize}[leftmargin=0.5cm]
\item[]Apply  Randomized Compiling to $\C^{\textup{(PEC)}}(\p^{(k)}_1,\ldots,\p^{(k)}_m)$. Apply the resulting circuit to the state $\rho_\textup{in}$, measure $O$ and save the result as $r_k$.
\end{itemize}

\textbf{Outputs.} \\
The number $\widehat{E}_\textup{PEC}(O)=C_\textup{tot}\sum\limits_{k=1}^N{s_kr_k}/{N}$.
\end{small}
}}
\end{center}

\vspace{0.6cm}

We can now prove theorem \ref{th:PEC}.

\begin{proof}\textit{(Theorem \ref{th:PEC}). }To prove that $\widehat{E}_\textup{PEC}(O)$ has standard deviation $O(\sigma)$ we use the same arguments as in \cite{TBG17}. Specifically, we begin by noting that each measurement outcome $r_k$ is an estimator of $\widehat{E}(O)$ with standard deviation $O(1)$, hence $C_\textup{tot}s_kr_k$ has standard deviation $O(C_\textup{tot})$. Therefore, $\widehat{E}_\textup{PEC}(O)$ has standard deviation $O(C_\textup{tot}/N^{1/2})$, which is $O(\sigma)$ for $N=(C_\textup{tot}/\sigma)^2$.

To calculate the bias of $\widehat{E}_\textup{PEC}(O)$, for simplicity we begin by considering the case $m=1$ and we generalize afterwards. For $m=1$, an ideal implementation of the input circuit performs the operation $\C=\E_2\h_1\E_1$. Averaging over all the random Pauli operators appended to the input circuit (both in step 2.1 for error cancellation and in step 3 for Randomized Compiling), the number $\widehat{E}_\textup{PEC}(O)$ equals
\begin{equation}
\label{eq:estimator_m=1_1}
    \widehat{E}_\textup{PEC}(O) = C_\textup{tot} \sum_{l=0}^{4^n-1}s_l\epsilon_l^{(\h_1)}\textup{Tr}\big[O\E_2\p_l\D_{\h_1}\h_1\E_1(\rho_\textup{in})\big]\:,
\end{equation}
where $\D_{\h_1}(\cdot)=\sum_{t=0}^{4^n-1}\epsilon_t^{(\h_1)}P_t(\cdot)P_t$ is a Pauli channel with error rates $\epsilon_l^{(\h_1)}$, $s_0=1$ and $s_l=-1$ for all $l\neq0$. The r.h.s. of the equation above can be rewritten as
\begin{align}
\label{eq:estimator_m=1_2}
\widehat{E}_\textup{PEC}(O) =& C_\textup{tot} \sum_{l,t=0}^{4^n-1}s_l\epsilon_l^{(\h_1)}\epsilon_t^{(\h_1)}\textup{Tr}\big[O\E_2\p_l\p_t\h_1\E_1(\rho_\textup{in})\big]\\
=&C_\textup{tot}\bigg(\sum_{l=0}^{4^n-1}s_l\big(\epsilon_l^{(\h_1)}\big)^2\:\textup{Tr}\big[O\E_2\h_1\E_1(\rho_\textup{in})\big]
+
\sum_{l,t\neq l}s_l\epsilon_l^{(\h_1)}\epsilon_t^{(\h_1)}\textup{Tr}\big[O\E_2\p_l\p_t\h_1\E_1(\rho_\textup{in})\big]
\bigg)\\
=&\textup{Tr}\big[O\C(\rho_\textup{in})\big]+
C_\textup{tot}\sum_{l,t\neq l}s_l\epsilon_l^{(\h_1)}\epsilon_t^{(\h_1)}\textup{Tr}\big[O\E_2\p_l\p_t\h_1\E_1(\rho_\textup{in})\big]
\:,
\end{align}
where in the last line we used $\sum_{l=0}^{4^n-1}
s_{l}\big(\epsilon_{l}^{(\h_1)}\big)^2=C^{-1}_\textup{tot}$. Thus,
\begin{align}
\delta_\textup{PEC}=&\bigg|\widehat{E}_\textup{PEC}(O)-\textup{Tr}\big[O\C(\rho_\textup{in})\big]\bigg|\\
=&C_\textup{tot}\bigg|\sum_{l,t\neq l}s_l\:\epsilon_l^{(\h_1)}\epsilon_t^{(\h_1)}\textup{Tr}\big[O\E_2\p_l\p_t\h_1\E_1(\rho_\textup{in})\big]\bigg|\\
=&C_\textup{tot}\bigg|\sum_{l\neq 0, t\neq0,l}\epsilon_l^{(\h_1)}\epsilon_t^{(\h_1)}\textup{Tr}\big[O\E_2\p_l\p_t\h_1\E_1(\rho_\textup{in})\big]
\bigg| \\
\leq & C_\textup{tot}\sum_{l\neq 0, t\neq0,l}\epsilon_l^{(\h_1)}\epsilon_t^{(\h_1)}\|O\|_\infty\|\E_2\p_l\p_t\h_1\E_1(\rho_\textup{in})\|_1 \tag{Triangle and H\"older's inequalities} \\
=&O\bigg(C_\textup{tot}\sum_{l\neq 0, t\neq0,l}\epsilon_l^{(\h_1)}\epsilon_t^{(\h_1)}\bigg)\:.
\end{align}
Using
\begin{equation}
\label{eq:epsilons}
\sum_{\substack{l\neq0\\t\neq 0,l}}\epsilon_l^{(\h_1)}\epsilon_t^{(\h_1)}
\leq
\sum_{\substack{l\neq0\\t\neq 0}}\epsilon_l^{(\h_1)}\epsilon_t^{(\h_1)}
=
\bigg(\sum_{\substack{l\neq0}}\epsilon_l^{(\h_1)}\bigg)^2
=
\big(1-\epsilon_{0}^{(\h_1)}\big)^2\:,
\end{equation}
we find $\delta_\textup{PEC}=O\big(C_\textup{tot}\big(1-\epsilon_{0}^{(\h_1)}\big)^2\big)$, which proves the theorem for $m=1$.

The generalisation to $m>1$ follows easily by linearity. Averaged over all the random Pauli operators, for $m>1$ the quantity $\widehat{E}_\textup{PEC}(O)$ returned by PEC equals
\begin{equation}
\label{eq:estimator_m>1_2}
\widehat{E}_\textup{PEC}(O)
=
C_\textup{tot}\sum_{l_1,\ldots,l_m}
s_{l_1}\cdots s_{l_m}
\epsilon_{l_1}^{(\h_1)}\cdots
\epsilon_{l_m}^{(\h_m)}\:
\textup{Tr}\bigg[O\E_{m+1}\p_{l_m}\D_{\h_m}\h_m\E_m\cdots\p_{l_1}\D_{\h_1}\h_1\E_1(\rho_\textup{in})\bigg]\:,
\end{equation}
where $s_{l_j}=1$ if $l_j=0$ and $s_{l_j}=-1$ otherwise. Following the same arguments as for $m=1$ we find
\begin{align}
    \delta_\textup{PEC}=&\bigg|\widehat{E}_\textup{PEC}(O)-\textup{Tr}\big[O\C(\rho_\textup{in})\big]\bigg|\\
    \approx&
    \bigg|
    \sum_{j=1}^mC_{\h_j}\sum_{l_j\neq0, t_j\neq0, l_j}\epsilon_{l_j}^{(\h_j)}\epsilon_{t_j}^{(\h_j)}
    \textup{Tr}\bigg[O\E_{m+1}\h_m\E_m\cdots\p_{l_j}\p_{t_j}\h_j\E_j\cdots\h_1\E_1(\rho_\textup{in})\bigg]
    \bigg|\:,
\end{align}
where we omitted terms that are of higher order in the Pauli error rates and we defined $C_{\h_j}=1/\big(\sum_{l_j}s_{l_j}\big(\epsilon_{l_j}^{(\h_j)}\big)^2\big)$. Using $C_{\h_j}\leq C_{\textup{tot}}$, Eq.~\ref{eq:epsilons} and the triangle inequality we finally find $\delta_\textup{PEC}=O\big(C_\textup{tot}\sum_{j=1}^{m}\big(1-\epsilon_{0}^{(\h_j)}\big)^2\big)$.
\end{proof}

\newpage
\begin{center}
\textbf{II. Formal description of NOX and proof of theorem 2}
\end{center}
We now provide a formal description of our NOX protocol, then we show a proof of theorem \ref{th:ce}. The formal description of NOX is given in the following box:

\vspace{0.6cm}

\begin{center}
\fbox{
    \parbox{0.8\textwidth}{
\begin{small}
\textbf{Box \hypertarget{CE}{2}.} Noiseless Output Extrapolation (NOX).\\
\noindent\rule{0.8\textwidth}{0.4pt}\\
\textbf{Inputs.} \\
An $n$-qubit quantum circuit $\C=\E_{m+1}\h_{m}\E_{m}\cdots\h_1\E_1$, an $n$-qubit state $\rho_\textup{in}$, an operator $O$ such that the spectral norm $||O||_\infty\sim 1$, a number $\sigma\in(0,1)$, an integer $\alpha>1$ and a Boolean $\textsf{id\_insert}\in\{\textsf{True, False}\}$.
\sbline
\textbf{Routine.}\\
\textbf{1.} Initialise the number $N=m^2/(\alpha-1)^2\sigma^2$.
\sbline
\textbf{2. }If \textsf{id\_insert} is \textsf{True}: 
\begin{itemize}
    \item[\textbf{2.1}] Apply the circuit $\C$ to $\rho_\textup{in}$ a total of $N$ times with Randomized Compiling. Estimate the expectation value $\widetilde{E}_\textup{in}(O)$ of $O$ at the end of the circuit.
    \item[\textbf{2.2}] Initialise a circuit $\C'_{\h_j,\alpha}$ by replacing $\h_j$ with $\h_j(\h_j\h_j^{-1})^{\alpha-1}$ in $\C$. Apply the circuit $\C'_{\h_j,\alpha}$ to $\rho_\textup{in}$ a total of $N$ times with Randomized Compiling. Estimate the expectation value $\widetilde{E}_{\h_j,\alpha}(O)$ of $O$ at the end of the circuit.
\end{itemize}
\textbf{\textcolor{white}{2. }}Else:
\begin{itemize}
    \item[\textbf{2.1}] Apply the circuit $\C$ to $\rho_\textup{in}$ a total of $N$ times with Randomized Compiling. Estimate the expectation value $\widetilde{E}_\textup{in}(O)$ of $O$ at the end of the circuit.
    \item[\textbf{2.2}] Initialise a circuit $\C'_{\h_j,\alpha}$ as in Eq.~\ref{eq:append_errors}. Apply the circuit $\C'_{\h_j,\alpha}$ to $\rho_\textup{in}$ a total of $N$ times with Randomized Compiling. Estimate the expectation value $\widetilde{E}_{\h_j,\alpha}(O)$ of $O$ at the end of the circuit.
\end{itemize}

\textbf{Outputs.} \\
The quantity $\widehat{E}_\textup{NOX}(O)=\widetilde{E}_\textup{in}(O)\frac{\alpha-1+m}{\alpha-1}-\sum_{j=1}^{m}\frac{\widetilde{E}_{\h_j,\alpha}(O)}{\alpha-1}$.
\end{small}
}}
\end{center}

We now prove theorem \ref{th:ce}.

\begin{proof}\textit{(Theorem \ref{th:ce}).}
We begin by proving that $\widehat{E}_\textup{NOX}(O)$ has standard deviation $O(\sigma)$. Since all the noisy estimators $\widetilde{E}_\textup{in}(O)$ and $\widetilde{E}_{\h_j,\alpha}(O)$ are calculated independently by running each circuit $N=m^2/(\alpha-1)^2\sigma^2$ times, they all have the same standard deviation $\sigma=O(N^{-1/2})$. Using the formula for propagation of errors, we find that the standard deviation of $\widehat{E}_\textup{NOX}(O)$ is
\begin{align}
O\bigg(\sqrt{\frac{1}{N}}
\sqrt{
\bigg(\frac{\alpha-1+m}{\alpha-1}\bigg)^2+
m\bigg(\frac{1}{\alpha-1}\bigg)^2
}\bigg)
=
O\bigg(\sqrt{\frac{(\alpha-1)^2\sigma^2}{m^2}}
\sqrt{\frac{m^2}{(\alpha-1)^2}}
\bigg)=O(\sigma)\:.
\end{align}
To calculate the bias, let us first define the maps $\delta_{\h_j}=\D_{\h_j}-\I$ for $j\in\{1,\ldots,m\}$. Note that $\delta_{\h_j}=(1-\epsilon_0^{(\h_j)})(\Q_j-\I)$ for some Pauli channel $\Q_j$. Hence, if the probability of error $1-\epsilon_0^{(\h_j)}$ is sufficiently small (as is the case for the  noisy cycles in state-of-the-art platforms) we have $||\delta_{\h_j}^2||_\diamond\ll||\delta_{\h_j}||_\diamond$, where $||\:\cdot\:||_\diamond$ represents the diamond norm. 

Defining $\delta_{\h_j}$ allows us to write $\widetilde{E}_\textup{in}(O)-E(O)$ as (see lemma \ref{lem:telescopic} for a proof)
\begin{align}
\label{eq:ce1}
\widetilde{E}_\textup{in}(O)-E(O)=\sum_{j=1}^mA_j\:,
\end{align}
where
\begin{equation}
    A_j=\textup{Tr}\bigg[O\E_{m+1}\bigg(\circ_{k=j+1}^m\h_k\E_k\bigg)\delta_{\h_j}\h_j\E_j\bigg(\circ_{k=1}^{j-1}\D_{\h_k}\h_k\E_k\bigg)(\rho_{\textup{in}})\bigg]\:.
\end{equation}
For every $j$, the quantity $A_j$ is a multiple of $\widetilde{E}_{\h_j,\alpha}(O)-\widetilde{E}_\textup{in}(O)$, modulo terms that are second order in $\delta_{\h_j}$. Indeed, since $\D_{\h_j}^\alpha=(\I+\delta_{\h_j})^\alpha=\I+\alpha\delta_{\h_j}+\alpha(\alpha-1)\delta_{\h_j}^2/2+O(\delta_{\h_j}^3)$ we have
\begin{align}
\label{eq:ce2}
    \widetilde{E}_{\h_j,\alpha}(O)-\widetilde{E}_\textup{in}(O)&=(\alpha-1)\:\textup{Tr}\bigg[O\E_{m+1}\bigg(\circ_{k=j+1}^m\D_{\h_k}\h_k\E_k\bigg)\delta_{\h_j}\h_j\E_j\bigg(\circ_{k=1}^{j-1}\D_{\h_k}\h_k\E_k\bigg)(\rho_{\textup{in}})\bigg]\\
    &= (\alpha-1)\big(A_j+B_j+C_j\big)\:,
\end{align}
where
\begin{align}
    B_j=&\sum_{l=j+1}^{m}\textup{Tr}\bigg[O\E_{m+1}\bigg(\circ_{k=l+1}^m\h_k\E_k\bigg)\delta_{\h_l}\h_l\E_l\bigg(\circ_{k=j+1}^{l-1}\h_k\E_k\bigg)\delta_{\h_j}\h_j\E_j\bigg(\circ_{k=1}^{j-1}\D_{\h_k}\h_k\E_k\bigg)(\rho_{\textup{in}})\bigg]\\
    &+\frac{\alpha}{2}\textup{Tr}\bigg[O\E_{m+1}\bigg(\circ_{k=j+1}^m\h_k\E_k\bigg)\delta^2_{\h_j}\h_j\E_j\bigg(\circ_{k=1}^{j-1}\D_{\h_k}\h_k\E_k\bigg)(\rho_{\textup{in}})\bigg]\nonumber\\
    =&O\bigg(\alpha(1-\epsilon_0^{\h_j})\sum_{l=j}^{m}(1-\epsilon_0^{\h_l})\bigg)
\end{align}
is quadratic in the $\delta$s and $C_j$ is cubic. Thus, combining Eq.s \ref{eq:ce1} and \ref{eq:ce2} we find
\begin{align}
    E(O)&=\widetilde{E}_\textup{in}(O)-\sum_{j=1}^m\frac{\widetilde{E}_{\h_j,\alpha}(O)-\widetilde{E}_\textup{in}(O)}{\alpha-1}+O\bigg(\alpha\sum_{j=1}^m(1-\epsilon_0^{\h_j})\sum_{l=j}^{m}(1-\epsilon_0^{\h_l})\bigg)\\
    &=\widehat{E}_\textup{NOX}(O)+O\bigg(\alpha\sum_{j=1}^m(1-\epsilon_0^{\h_j})\sum_{l=j}^{m}(1-\epsilon_0^{\h_l})\bigg)
\end{align}
and finally
\begin{equation}
    \delta_\textup{NOX}(O)=\big|{E}(O)-\widehat{E}_\textup{NOX}(O)\big|=O\bigg(\alpha\sum_{j=1}^m(1-\epsilon_0^{\h_j})\sum_{l=j}^{m}(1-\epsilon_0^{\h_l})\bigg)\:.
\end{equation}
\end{proof}

 We end the section by proving the following lemma, which we use to derive Eq.~\ref{eq:ce1}.
\begin{lemma}
\label{lem:telescopic}
Let $\Phi_m=\circ_{j=1}^m\U_j$, $\widetilde{\Phi}_m=\circ_{j=1}^m\D_{j}\U_j$ and $\D_j=I+\delta_j$ for all $j\in\{1,\ldots,m\}$. The following equality holds:
\begin{equation}
\label{eq:lemma}
    \widetilde{\Phi}_m=\Phi_m+\sum_{j=1}^m\bigg(\circ_{k=j+1}^m\U_k\bigg)\delta_j\U_j\bigg(\circ_{k=1}^{j-1}\D_k\U_k\bigg)\:.
\end{equation}
\end{lemma}
\begin{proof}\textit{(Lemma \ref{lem:telescopic}).}
We prove the lemma using induction. Eq.~\ref{eq:lemma} holds trivially for $m=1$, since the l.h.s equals $\D_1\U_1$ and the r.h.s. equals $\U_1+\delta_1\U_1=\D_1\U_1$. To complete the induction we now assume that Eq.~\ref{eq:lemma} holds for a given $m>1$ and we prove the equality for $m+1$. For $m+1$ the r.h.s. of Eq.~\ref{eq:lemma} can be rewritten as
\begin{align}
    &\Phi_{m+1}+\sum_{j=1}^{m+1}\bigg(\circ_{k=j+1}^m\U_k\bigg)\delta_j\U_j\bigg(\circ_{k=1}^{j-1}\D_k\U_k\bigg)\nonumber\\
    =&
    \Phi_{m+1}+
    \delta_{m+1}\U_{m+1}\bigg(\circ_{k=1}^{m}\D_k\U_k\bigg)
    +\U_{m+1}\sum_{j=1}^{m}\bigg(\circ_{k=j+1}^m\U_k\bigg)\delta_j\U_j\bigg(\circ_{k=1}^{j-1}\D_k\U_k\bigg)\\
    =&
    \Phi_{m+1}+
    \delta_{m+1}\U_{m+1}\widetilde{\Phi}_m
    +\U_{m+1}\widetilde{\Phi}_m-\U_{m+1}\Phi_m\\\cr
    =&
    \delta_{m+1}\U_{m+1}\widetilde{\Phi}_m
    +\U_{m+1}\widetilde{\Phi}_{m}\\\cr
    =&\D_{m+1}\U_{m+1}\widetilde{\Phi}_{m}\:.
\end{align}
Since $\D_{m+1}\U_{m+1}\widetilde{\Phi}_{m}=\widetilde{\Phi}_{m+1}$, the equalities above show that if Eq.~\ref{eq:lemma} holds for a given $m$, then it also holds for $m+1$. This proves the lemma.
\end{proof}

\begin{center}
\textbf{III. Proof of Lemmas 1 and 2}
\end{center}
In this section we prove our lemmas \ref{lem:relax_PEC} and \ref{lem:relax_ce}. We begin by proving lemma \ref{lem:relax_PEC}.
\begin{proof}\textit{(Lemma \ref{lem:relax_PEC})}
We begin by analysing the performance of PEC. To do so, let us denote as 
\begin{equation}
    \U=C_\textup{tot}\E_{m+1}\bigg(\circ_{j=1}^m\R_{\h_j}\D _{\h_j}\h_j\E_j\bigg)
\end{equation}
the map implemented on average by the PEC circuits when the Pauli error rates are known exactly, with $\R_{\h_j}=\sum_{l=0}^{4^n-1}s_l\epsilon_l^{(\h_j)}\p_l$. Similarly, let us denote as
\begin{equation}
    \U'=C_\textup{tot}\E_{m+1}\bigg(\circ_{j=1}^m{\R}'_{\h_j}\D _{\h_j}\h_j\E_j\bigg)
\end{equation}
the map implemented on average by the PEC circuits when the noise reconstruction is inaccurate, with ${\R}'_{\h_j}=\sum_{l=0}^{4^n-1}s_l\widehat{\epsilon}_l^{(\h_j)}\p_l$. The estimator $\widehat{E}_\textup{PEC}(O)$ returned by PEC can be rewritten as
\begin{align}
    \widehat{E}_\textup{PEC}(O)&= \textup{Tr}\big[O\U'(\rho_\textup{in})\big]\\
    &=\textup{Tr}\big[O\U(\rho_\textup{in})\big]+\textup{Tr}\big[O(\U'-\U)(\rho_\textup{in})\big]\\
    &\leq\textup{Tr}\big[O\U(\rho_\textup{in})\big]+\|O\|_{\infty}\|(\U'-\U)(\rho_\textup{in})\|_1 \tag{H\"older's inequality} \\
    &\leq \textup{Tr}\big[O\U(\rho_\textup{in})\big]+\|O\|_{\infty}||\U-\U'||_\diamond\:,
\end{align}
where $||\cdot||_\diamond$ denotes the diamond norm. Using the same arguments as in theorem 2 in Ref.~\cite{WE16}, we find
\begin{equation}
    ||\U-\U'||_\diamond\leq C_\textup{tot}\sum_{j=1}^m||{\R}_{\h_j}-{\R}'_{\h_j}||_\diamond\leq C_\textup{tot}\sum_{j=1}^m\sum_l\big|{\epsilon}_l^{(\h_j)}-\widehat{\epsilon}_l^{(\h_j)}\big|\:,
\end{equation}
which together with theorem \ref{th:PEC} proves the lemma.
\end{proof}

We can now prove lemma \ref{lem:relax_ce}.
\begin{proof}\textit{(Lemma \ref{lem:relax_ce})}
Let 
\begin{equation}
    \U_{\h_j,\alpha}=\E_{m+1}\D_{\h_m}\h_m\E_m\cdots\D_{\h_j}^{\alpha}\h_j\cdots\E_2\h_1\E_1
\end{equation}
be the map implemented by a circuit with noise on the $j$th  noisy cycle amplified perfectly, and let 
\begin{equation}
    \U_{\h_j,\alpha}'=\E_{m+1}\D_{\h_m}\h_m\E_m\cdots\widetilde{\R}_{\h_j}\D_{\h_j}\h_j\cdots\E_2\h_1\E_1
\end{equation}
be the map implemented by a circuit with noise on the $j$th  noisy cycle amplified imperfectly. The estimator $\widehat{E}_\textup{NOX}(O)$ returned by NOX can be written as
\begin{align}
  \widehat{E}_\textup{NOX}(O) &= \widetilde{E}_\textup{in}(O)  \frac{\alpha-1+m}{\alpha-1}-\frac{1}{\alpha-1}\sum_{j=1}^m\textup{Tr}\big[O\U_{\h_j,\alpha}'(\rho_\textup{in})\big]\\
  &= \widetilde{E}_\textup{in}(O)  \frac{\alpha-1+m}{\alpha-1}-\frac{1}{\alpha-1}\sum_{j=1}^m\textup{Tr}\big[O\U_{\h_j,\alpha}(\rho_\textup{in})\big]-
  \frac{1}{\alpha-1}\sum_{j=1}^m\textup{Tr}\big[O(\U_{\h_j,\alpha}'-\U_{\h_j,\alpha})(\rho_\textup{in})\big]\\
    &\leq \widetilde{E}_\textup{in}(O)  \frac{\alpha-1+m}{\alpha-1}-\frac{1}{\alpha-1}\sum_{j=1}^m\widetilde{E}_{\h_j}(O) +
  \frac{1}{\alpha-1}\sum_{j=1}^m \|O\|_{\infty}\|(\U_{\h_j,\alpha}'-\U_{\h_j,\alpha})(\rho_\textup{in})\|_1, \tag{H\"older's inequality}\\
  &\leq \widetilde{E}_\textup{in}(O)  \frac{\alpha-1+m}{\alpha-1}-\frac{1}{\alpha-1}\sum_{j=1}^m\widetilde{E}_{\h_j}(O) +
  \frac{\|O\|_\infty}{\alpha-1}\sum_{j=1}^m||\widetilde{\R}_{\h_j}-\D_{\h_j}^{\alpha-1}||_\diamond\:,
\end{align}
which together with theorem \ref{th:ce} proves the lemma.
\end{proof}

\end{document}